\documentclass[sigconf,screen]{acmart}
\AtBeginDocument{%
  }

\setcopyright{cc}
\setcctype{by-nc-nd}
\acmDOI{10.1145/3832783.3837496}
\acmYear{2026}
\copyrightyear{2026}
\acmISBN{979-8-4007-2882-2/2026/10}
\acmConference[ASE '26]{Proceedings of the 41st IEEE/ACM International Conference on Automated Software Engineering}{October 12--16, 2026}{Munich, Germany}
\acmBooktitle{Proceedings of the 41st IEEE/ACM International Conference on Automated Software Engineering (ASE '26), October 12--16, 2026, Munich, Germany}
\acmSubmissionID{ase26main-p1456-p}
\usepackage{algorithm}
\usepackage{algorithmic}
\usepackage{proof}
\usepackage{listings}
\usepackage{subcaption}
\usepackage{mathpartir}
\usepackage{makecell}
\usepackage{multirow}
\usepackage{booktabs}
\usepackage{booktabs}
\usepackage{longtable}
\usepackage{wrapfig}

\newcommand{\ininso}{o_{\ell^\prime}^+}
\newcommand{\indelo}{o_{\ell^\prime}^-}

\newcommand{\pts}{\mathcal{P}}
\newcommand{\cg}{\mathcal{E}}
\newcommand{\OUT}{\texttt{OUT}}
\newcommand{\IN}{\texttt{IN}}
\newcommand{\KILL}{\texttt{KILL}}
\newcommand{\GEN}{\texttt{GEN}}
\newcommand{\avgfull}{9.60}
\newcommand{\avgreset}{5.84}
\newcommand{\avgbuildcg}{49.89}
\newcommand{\avgsilva}{15.8\%}
\newcommand{\avgscc}{3.7\%}
\newcommand{\avgnodes}{51.9\%}

\usepackage{xcolor}
\usepackage{enumitem}
\newcommand{\red}[1]{#1}

\definecolor{githubred}{RGB}{255,200,200}
\definecolor{githubgreen}{RGB}{200,255,200}
\newcommand{\highlight}[2]{\colorbox{#2}{#1}}
\graphicspath{{./}{./images/}}
\setlist{noitemsep, topsep=2pt}
\usepackage{etoolbox}

\begin{document}
\title{IncSFS: Incremental Full-Sparse Flow-Sensitive Pointer Analysis for C/C++}

\author{Kunlin Liu}
\authornote{Also with the affiliation: State Key Laboratory of Complex \& Critical Software Environment, National University of Defense Technology, Changsha, China.}
\orcid{0009-0003-4384-430X}
\affiliation{%
	\department{College of Computer Science and Technology,}
	\institution{National University of Defense Technology}
	\city{Changsha}
	\country{China}
}
\email{klliu18@nudt.edu.cn}
\author{Zhenbang Chen}
\authornotemark[1]
\authornote{Zhenbang Chen is the corresponding author.}
\orcid{0000-0002-4066-7892}
\affiliation{%
	\department{College of Computer Science and Technology,}
	\institution{National University of Defense Technology}
	\city{Changsha}
	\country{China}
}
\email{zbchen@nudt.edu.cn}
\author{Piyi Zu}
\authornotemark[1]
\orcid{0009-0001-6680-3602}
\affiliation{%
	\department{College of Computer Science and Technology,}
	\institution{National University of Defense Technology}
	\city{Changsha}
	\country{China}
}
\email{piyizu@nudt.edu.cn}
\author{Yide Du}
\authornotemark[1]
\orcid{0000-0003-0372-7101}
\affiliation{%
	\department{College of Computer Science and Technology,}
	\institution{National University of Defense Technology}
	\city{Changsha}
	\country{China}
}
\email{dyd1024@nudt.edu.cn}
\author{Ji Wang}
\authornotemark[1]
\orcid{0000-0003-0637-8744}
\affiliation{%
	\department{College of Computer Science and Technology,}
	\institution{National University of Defense Technology}
	\city{Changsha}
	\country{China}
}
\email{wj@nudt.edu.cn}

\begin{abstract}
	Pointer analysis is the fundamental technique for compiler optimization and program analysis. Flow sensitive pointer analysis provides high precision but hard to scale to large-size projects. Tailored for rapid iteration scenarios where software evolves continuously, we introduce IncSFS, the first incremental full-sparse flow-sensitive pointer analysis algorithm for C/C++ programs. The algorithm contains two main steps. It first transforms the value-flow graph of a program into a constraint graph, on which a strongly-connected component detection is performed to ensure precision. It then propagates increases and decreases in points-to set in an interleaving manner to support both code deletion and insertion within a single analysis pass. \red{IncSFS is guaranteed to terminate and compute the least fixed point when point-to relation during the analysis is object-acyclic.}
	Experimental results on six large-scale real-world projects show that IncSFS is both precise and efficient, achieving average speedups of \avgfull x over full flow-sensitive pointer analysis and of \avgreset x over the traditional reset-recompute incremental approach. Furthermore, IncSFS improves efficiency by \avgsilva\ compared with state-of-the-art incremental pointer analysis algorithms that also propagate changes in points-to sets.
\end{abstract}


\maketitle

\section{Introduction}
Pointer analysis statically computes the memory address which pointer variables point to at runtime. Its precision is crucial to compiler optimization and program analysis.  One of precision improvement dimension is flow-sensitive. It provides high precision  and serves as a foundational technique in a wide array of upstream applications, ranging from compiler optimizations~\cite{966497,10.1145/277650.277670,10.1145/1250734.1250766,10.1145/268946.268957,10.1145/781498.781502} to complex program analyses such as typestate analysis~\cite{10.1145/1348250.1348255}, multi-threaded program analysis~\cite{10.1145/379539.379553,10.1145/2854038.2854043}, bug detection~\cite{1760267.1760284,3241189.3241268}, and program slicing~\cite{10.1145/1250734.1250748,li_et_al:LIPIcs.ECOOP.2016.15}. 

Compared with flow-insensitive which treats statements as unordered, flow-sensitive pointer analysis follows the program’s execution order in control flow graph (CFG).  As program executes, the memory addresses pointed to by a pointer variable may change over time and differs at different statements. To capture this dynamic behavior, flow-sensitive pointer analysis propagates points-to information statement by statement and maintains a separate points-to information for each pointer variable at every program statement. It leverages \emph{strong update} ~\cite{10.1145/1925844.1926389,10.1007/978-3-642-11957-6_14} to update points-to information with a new one if we know the points-to information of this pointer variable must be updated. Since real-world programs often comprise hundreds of thousands of pointer variables and memory objects, and the program can have hundreds of thousands of statements, flow-sensitive pointer analysis struggles to scale to large-size software \cite{10.1145/347324.348916,5764696,647166.717860,10.1007/11688839_3}. Due to its practical importance, lots of works have been proposed to improve its scalability~\cite{10.1145/1594834.1480911,5764696,9370334,10.1007/11688839_3,10.1145/1772954.1772985,10.1145/2025113.2025160,10.1145/1925844.1926389}. The state-of-the-art and most widely used  flow-sensitive pointer analysis is \textbf{full-sparse flow-sensitive pointer analysis (SFS)}~\cite{5764696}. The primary efficiency gain of SFS stems from its use of \emph{value-flow graph} (VFG) rather than CFG, where a variable’s points-to set is propagated directly from its definition sites to potential use sites, skipping irrelevant statements. Despite its significantly improved efficiency, \emph{SFS faces the same scalability challenges as traditional flow-sensitive pointer analyses and still struggles to scale to large codebases—a long-standing open problem}.

\par This scalability bottleneck is particularly critical due to the widespread adoption of Continuous Integration and Continuous Delivery (CI/CD) in modern software engineering. To meet fast-feedback requirements, incremental pointer analysis becomes crucial by updating results based on code changes rather than recomputing them from scratch. In recent years, a growing body of approaches have been proposed for incremental pointer analysis~\cite{10.1145/3293606,10.1145/3527332,10.1145/3786763,10.1145/3720436,10.1007/978-3-031-24950-1_14,10.1145/3725214}.  These approaches  can be divided into three categories: 1) The \emph{Reset–Recompute} approach~\cite{10.1145/170035.170066,841034,10.1145/3786763}, which resets the points-to results at all affected statements to a safe initial estimate ,\emph{i.e.}, an empty set and recomputes them via a full re-analysis. \red{In particular, the work by Yur \emph{et al.} was the first to apply the Reset--Recompute approach to incremental traditional flow- and context-sensitive pointer analysis, rather than to full-sparse flow-sensitive pointer analysis.} 2) The \emph{Restart–Iteration} approach~\cite{10.1145/502874.502898,ghodssi1983incremental,10.1007/978-3-031-24950-1_14}, which resets the points-to results at only the modified statements to an empty set and recomputes. 3) The \emph{Incremental Pointer Analysis} (IPA) framework~\cite{10.1145/3293606,10.1145/3527332,10.1145/3725214}, which computes and propagates only the \emph{deltas} (differences) in points-to information instead of resetting.%

\par \textbf{Limitation of existing approaches.} The first category ensures the safety of results but faces performance problem. The process of resetting and recomputing often leads to significant redundant computations. The second category enjoys efficiency but suffers from  imprecision when reusing old analysis results within strongly connected components (SCCs) of the program~\cite{10.1145/3720436}. Both the incremental flow-sensitive approaches above operate directly on the control-flow graph rather than the full-sparse \emph{value-flow graph}.  The third category achieves high efficiency by leveraging delta propagation but is currently limited to flow-insensitive settings. As a result, \emph{no prior researches on incremental analysis have been conducted on the full-sparse flow-sensitive pointer analysis.}  %

\par
In this paper, we present the \textbf{first incremental algorithm for full-sparse flow-sensitive pointer analysis}. Inspired by IPA, we aim to propagate the deltas of points-to sets from code modification as it is particularly well suited for pointer analysis and has been shown to be more efficient than reset-based methods~\cite{10.1145/3293606}.

\begin{figure}
	\centering
	\begin{minipage}{0.6\linewidth} \begin{lstlisting}
|p = x; // p  and x points to \{o1,o2\}|
|\highlight{- p = y; // y points to \{o2,o3\}}{githubred}|
|q = p; //
		\end{lstlisting}
		\caption{IPA performs inefficiently in this deletion scenario}
		\label{fig:intro-example}
	\end{minipage}
\end{figure}

\par \textbf{Challenges.} IPA is inherently tailored to flow-insensitive analysis, where constraint edges encode propagation paths between variables. When one propagation path is removed, IPA inspects the remaining propagation path from other variables to validate whether the potentially removed memory objects should be removed. For an example, in Figure \ref{fig:intro-example}, q originally points-to \{o1,o2,o3\} which is propagated from x and y in a flow-insensitive way. When Line 2 is deleted, the propagation path from y to q is deleted and IPA find another propagation path from x still contain o2, then only $o3$ is removed from points-to set of q. In IPA, all propagation paths require to be found, otherwise resulting in incorrectness. However, flow-sensitive analysis is fundamentally more complex. To find propagation paths between variables in flow-sensitive pointer analysis, IPA further requires to consider the  control flow reachability between variables and  ensure points-to set of variables is not strongly updated in the propagation process. In this example, x cannot flow to q since p is strongly  updated at Line 2.  Since the strong update could happen for any variable at any statement, different variables may have different propagation path. How to encode propagation paths with control flow information for all variables is the \textbf{key challenge} for extending IPA to flow-sensitivity. 

Another \textbf{key challenge} is how to characterize the evolution of the SFS-based value flow graph in response to code changes. \red{As illustrated in Fig.\ref{fig:intro-example}, when the second line is deleted, the value flow edge from \texttt{p = y} to \texttt{q = p} is deleted and \texttt{p = x} to \texttt{q = p} is added.} The points-to set of \texttt{q} changes from $\{o2, o3\}$ to $\{o1, o2\}$, resulting in the removal of $o3$ and the addition of $o1$. However, traditional IPA typically assumes that code deletion leads solely to set reduction, thus lacking an efficient and unified algorithm to handle such complex, non-monotonic updates. Existing approaches struggle to identify the correlation between additions and deletions, thus resorting to  a "delete-then-add" engineering strategy~\cite{10.1145/3293606,10.1145/3527332,10.1145/3725214}.

\textbf{Our solution.} To encode control flow information and align with IPA, our insight is to instantiate a unique renamed variable for the variable at the redefined statement to capture its local points-to set.  We then construct a constraint graph used in flow-insensitive analysis to propagates points-to set among these renamed variables. Consequently, we perform IPA on the constraint graph in a flow-insensitive way.
To address the second challenge, we compute increases and decreases in points-to sets in an interleaved manner in response to value flow graph changes. To guarantee termination and convergence to the least fixed point, we ensure that increases and decreases are mutually exclusive and require object-acyclicity throughout the analysis: the points-to relation remains acyclic. Another key scenario is code \emph{replacement}. Although Git represents replacements as deletions followed by insertions, the statements are often semantically correlated; our method exploits this correlation to avoid redundant propagation in standard IPA.

Based on the above approaches, we have designed our incremental full-sparse flow-sensitive pointer analysis, \emph{i.e.}, IncSFS, and implemented it in the SVF framework~\cite{10.1145/2892208.2892235}. %
We have evaluated our approach in large scale real-world programs. We compare our approach with both full analysis and traditional reset-recompute approach, and it  achieves \avgfull x and \avgreset x  speedups respectively. We also compare our interleaved propagation with IPA's separate propagation and it achieves \avgsilva\ efficient improvement.
\par
In summary, we claim the following contributions:
\begin{itemize}[leftmargin=12pt,itemsep=2pt,topsep=2pt]
	\item We propose the first incremental algorithm for full-sparse flow-sensitive pointer analysis that efficiently updates the points-to information  in CI/CD scenarios.
	
	\item We \red{rename variables to enable safe propagation of points-to deltas in incremental settings by transforming the VFG into a variable-level constraint graph}.
	
	\item We support both code insertions and deletions by jointly propagating positive and negative deltas; \red{under the object-acyclicity condition, the algorithm terminates and reaches the least fixed point despite the non-monotonic effects of code modification.}
	
	\item We evaluate our approach on large-scale real-world programs and demonstrate that it achieves \avgfull x and \avgreset x  speedups than full analysis and reset-recompute approach respectively.
\end{itemize}
\begin{figure}[t]
	\centering
	
	\begin{minipage}[t]{0.48\textwidth} 
		\centering
		\footnotesize
		\captionof{table}{Program Syntax} %
		\label{tab:program-syntax}
		\vspace{0pt} 
		\begin{tabular}{ll}
			\toprule
			Type & Statement \\ 
			\midrule
			Alloc & $p = alloc_o$ \\
			Copy  & $p = q$ \\
			Field & $p = \&q.f$ \\
			Phi   & $p = \phi(p_1,\dots,p_n)$ \\
			Store & $*p = q$ \\
			Load  & $p = *q$ \\
			\bottomrule
		\end{tabular}
	\end{minipage}%
	\hspace{0.5cm} %
	\begin{minipage}{0.45\textwidth}
		\centering
		\footnotesize
		\captionof{table}{ Abstract Domain}
		\label{tab:abstract-domain}
		\begin{tabular}{c c}
			\toprule
			Notation & Definition \\
			\midrule
			$\mathcal{L}$ & Labels \\
			$\mathcal{O}$ & Address Taken Variables\\
			$\mathcal{T}$ & Top Level Variables\\
			$\mathcal{V}:=\mathcal{O}\cup\mathcal{T}$ & Variables \\
			$\mathcal{P}:=\mathcal{V}\rightarrow2^{\mathcal{O}}$ & Points-to Set of Variables\\
			$\IN:=\mathcal{L}\rightarrow\mathcal{V}\rightarrow2^{\mathcal{O}}$ & Points-to set at each label\\
			$\OUT:=\mathcal{L}\rightarrow\mathcal{V}\rightarrow2^{\mathcal{O}}$ & Points-to set at each label\\
			\bottomrule
		\end{tabular}
	\end{minipage}
	
\end{figure}

\section{PRELIMINARIES And Motivation}

\subsection{Program Syntax}

\begin{figure*}[t]
	\centering

	\begin{minipage}[t]{0.38\textwidth}
		\vspace{-15mm}
		\centering

		\begin{minipage}[t]{0.40\linewidth}
			\raggedright
			\begin{lstlisting}[basicstyle=\ttfamily\footnotesize, xleftmargin=0pt, escapeinside={||}]
*x = a;
|\highlight{-*x = b;}{githubred} \label{line:strongupdate}|
c = *x;
*q = b;
*p = d;
while(...){
	*m = a;
	*n = c;
}
			\end{lstlisting}
		\end{minipage}%
		\hspace{-10mm}
		\begin{minipage}[t]{0.40\linewidth}
			\centering
			\vspace{0pt}
			\footnotesize
			\setlength{\tabcolsep}{1.5pt} %
			\begin{tabular}{cl}
				\toprule
				Var & Pts Set \\
				\midrule
				x & $\{o^x\}$\\
				p & $\{o^p\}$ \\
				q & $\{o^q\}$ \\
				m & $\{o^q\}$ \\
				n & $\{o^p,o^q\}$ \\
				a & $\{\red{o^a}\}$ \\
				b & $\{\red{o^a},\red{o^b}\}$ \\
				d & $\{\red{o^a}\}$ \\
				\bottomrule
			\end{tabular}
		\end{minipage}
		\vspace{2pt}
		\captionof{figure}{Motivating example.}
		\label{fig:example}
	\end{minipage}
	\begin{minipage}[t]{0.60\textwidth}
		\centering
		\begin{minipage}[c]{0.50\textwidth}
			\centering

			\includegraphics[width=\linewidth]{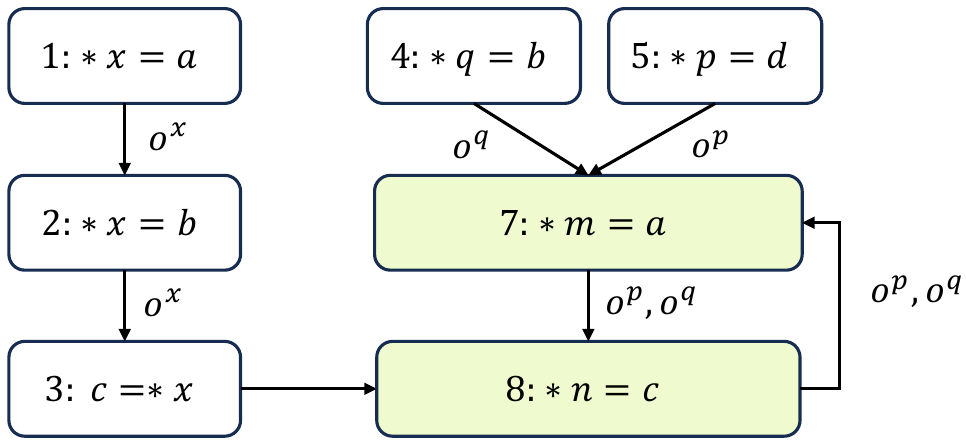}

		\end{minipage}%
		\hspace{-8mm}
		\begin{minipage}[c]{0.50\textwidth}
			\centering
			\footnotesize
			\renewcommand{\arraystretch}{1.2}

			\begin{tabular}{ccll}
				\toprule
				Node & Variable & $\IN$ & $\OUT$ \\
				\midrule
				1 & $o^x$ & $\varnothing$ & $\{\red{o^a}\}$ \\
				2 & $o^x$ & $\{\red{o^a}\}$     & $\{\red{o^a},\red{o^b}\}$ \\
				3 & $o^x$ & $\{\red{o^a},\red{o^b}\}$ & $\varnothing$ \\
				4 & $o^q$ & $\varnothing$ & $\{\red{o^a},\red{o^b}\}$ \\
				5 & $o^p$ & $\varnothing$ & $\{\red{o^a}\}$ \\
				\multirow{2}{*}{7} & $o^p$ & $\{\red{o^a},\red{o^b}\}$ & $\{\red{o^a},\red{o^b}\}$ \\
				& $o^q$ & $\{\red{o^a},\red{o^b}\}$ & $\{\red{o^a}\}$ \\
				\multirow{2}{*}{8} & $o^p$ & $\{\red{o^a},\red{o^b}\}$ & $\{\red{o^a},\red{o^b}\}$ \\
				& $o^q$ & $\{\red{o^a}\}$ & $\{\red{o^a},\red{o^b}\}$ \\
				\bottomrule
			\end{tabular}
		\end{minipage}
		\caption{the VFG (left) and the corresponding points-to sets during the sparse flow-sensitive analysis (right).}
		\label{fig:motivating_example}
	\end{minipage}
\end{figure*}

\par Table~\ref{tab:program-syntax} defines our program syntax with six statement forms: \textbf{Alloc}, \textbf{Copy}, \textbf{Field}, \textbf{Phi}, \textbf{Store}, and \textbf{Load}. The input programs are in a static single assignment (SSA) form, where each variable is defined only once. %
The input program is in a static single assignment (SSA) form in which variable  is defined only once.

\subsection{Full-Sparse Flow-Sensitive Pointer Analysis}

\subsubsection{Abstract Domains.} The abstract domains are illustrated in Table \ref{tab:abstract-domain}. We  assign a label $\ell\in \mathcal{L}$ for each statement, representing the location of statement within the program. All Variables are partitioned as $\mathcal{V}=\mathcal{T}\cup\mathcal{O}$, where $\mathcal{T}$ contains contains top-level SSA variables and $\mathcal{O}$ contains address-taken variables, such as abstract heap objects accessed through loads and stores. For example, as shown in Figure \ref{fig:example}, we omit the definitions of variables \texttt{a}, \texttt{b}, \texttt{d}, \texttt{p}, \texttt{q}, \texttt{m}, \texttt{n}, \texttt{x} and assume the points-to sets of these variables are shown in  Figure \ref{fig:example} on the right. Variables \texttt{x}, \texttt{a} , \texttt{b}, \texttt{c}, \texttt{d}, \texttt{m}, \texttt{n}, \texttt{p}, \texttt{q} are top level variables while $o^p,o^q,o^x,\red{o^a},\red{o^b}$ are address-taken variables. $\mathcal{P}$ is a mapping which stores the points-to set for each top-level variable. Flow-sensitive pointer analysis maintains  $\IN$ and $\OUT$ at each statement which maps every  variable to its points-to set before and after that statement.

\subsubsection{Full-Sparse Flow-Sensitive Pointer Analysis}

Traditional flow-sensitive pointer analysis blindly computes and propagates points-to result of all variable statement by statement along the CFG regardless of it is used or not in the statement, resulting in inefficiency. Full-sparse flow-sensitive pointer analysis (SFS)~\cite{5764696} improve this by eliminating unnecessary points-to information propagation of a variable if this variable is not used at this statement. For example, since $o^x$ is not used at Line 4, SFS does not maintain points-to result of $o^x$ in $\IN_4$ and $\OUT_4$ and does not propagate the points-to set of $o^x$ to Line 4. To realize it,  SFS leverages a value-flow graph (VFG) (Definition \ref{def:vfg}) to propagate points-to sets directly from definition sites to use sites~\cite{10.1145/1594834.1480911, 5764696, 10.1145/2345156.2254092}.

\begin{definition}[\textbf{Value-Flow Graph}]
	\label{def:vfg} Value-flow graph (VFG) is defined as a labeled directed graph  $G= (N,E)$, where $N$ is a set of nodes, each of which represents a program statement and $E\subseteq N\times N\times\mathcal{V}$ is a set of labeled directed edges.  Each edge  $n_1\stackrel{o}{\rightsquigarrow}n_2\in E$ indicates that the variable $o$ defined at statement $n_1$ may be used at statement $n_2$.
\end{definition}
\par
\textbf{Construction of the VFG.}  SFS relies on  a flow-insensitive pre-analysis to identify address-taken variables potentially accessed by load and store statements. It then treats stores as potential definitions and loads as potential uses and connects their def-use relations to construct the VFG. As shown in Figure \ref{fig:example}, $\pts(m)$ here is over-approximated as $\{o^p,o^q\}$ under the pre-analysis, while the other points-to sets remain unchanged. Thus, \texttt{*m = a} may define both $o^p$ and $o^q$, and \texttt{*n = c} may use and redefine them, producing a value-flow edge from \texttt{*m = a} to \texttt{*n = c} labeled with $o^p$ and $o^q$.

\textbf{Full-sparse flow-sensitive pointer analysis algorithm.}   SFS propagates points-to result along the VFG instead of CFG following Equation~\ref{dataflowequation}, $\IN_n(o)$ is computed as the union of all $\OUT_k(o)$ from predecessor nodes $k$ that can reach $n$ via a value-flow edge, \emph{i.e.}, $k \stackrel{o}{\rightsquigarrow} n$. The $\OUT_n(o)$ set is then strongly updated by combining newly generated pointers $\GEN_n(o)$ with those from $\IN_n(o)$ that are not killed by $n$ i.e., $\IN_n(o) \setminus \KILL_n(o)$. Three behaviors are categorized based on $\pts(p)$ for store statement $\ell: *p = q$: no operation (Definition \ref{def:noop}), \emph{strong update} (Definition \ref{def:su}), and \emph{weak update} (Definition \ref{def:wu}). Taking Line 2: \texttt{*x = b} as an example, since \texttt{x} only point to $o^x$, a strong update is performed: it kills all incoming points-to set from the predecessor ($\KILL_2(o^x)=\IN_2(o^x)$) and generate a new set for $o^x$: $\OUT_2(o^x)=\pts(b)$. Since  \texttt{*n = c} where \texttt{n} points to two objects,  a weak updates is performed and $\forall _{o\in \pts(n)}\KILL_8(o)=\emptyset$.
\begin{equation}
	\label{dataflowequation}
	\begin{gathered}
		\IN_n(o) = \bigcup_{k\stackrel{o}{\rightsquigarrow} n} \OUT_k(o)\\
		\OUT_n(o)= \GEN_n(o)\cup (\IN_n(o)\setminus \KILL_n(o))
	\end{gathered}
\end{equation}

\begin{definition}[\textbf{No Operation}]
	For store statement $\ell: *p = q$, if $\vert\pts(p)\vert=0$, SFS does not handle this statement to ensure monotonicity and  $\forall_{o_\in\mathcal{O}}\OUT_\ell(o)=\emptyset$.
	\label{def:noop}
\end{definition}

\begin{definition}[\textbf{Strong Update}]
	For store statement $\ell: *p = q$, if $\vert\pts(p)\vert=1\wedge o\in\pts(p)$, SFS replaces $\OUT_\ell(o)$ with a new value,\emph{i.e.}, $\OUT_\ell(o)=\pts(q)$. And for other address-taken variables $o^\prime$ ($o\neq o^\prime$) reaching this statement, the points-to result just bypass this node,\emph{i.e.},$\OUT_\ell(o^\prime)=\IN_\ell(o^\prime)$.
	\label{def:su}
\end{definition}

\begin{definition}[\textbf{Weak Update}]
For store statement $\ell: *p = q$,	 if $\vert\pts(p)\vert>1\wedge o\in\pts(p)$, SFS merges $\OUT_\ell(o)$ with a new value ,\emph{i.e.}, $\OUT_\ell(o)=\IN_\ell(o)\cup\pts(q)$. And for other address-taken variables $o^\prime$ reaching this statement, the points-to result just bypass this node,\emph{i.e.},$\OUT_\ell(o^\prime)=\IN_\ell(o^\prime)$.
	\label{def:wu}
\end{definition}
In the presence of loops or recursive calls, the analysis iteratively updates the points-to sets until a fixed point is reached \cite{10.1145/1594834.1480911}. For example in Figure \ref{fig:motivating_example}, the nodes shaded in green form an SCC. We present a walkthrough of SFS starting with Node 8. Initially the $\IN$ and $\OUT$ at node 7 and node 8 are all $\emptyset$. We firstly focus on node 8: since $\pts(c)=\IN_3(o^x)=\OUT_2(o^x)=\{\red{o^a},\red{o^b}\}$, then $\OUT_8(o^p)=\IN_8(o^p)\cup\pts(c)=\emptyset \cup \{\red{o^a},\red{o^b}\}=\{\red{o^a},\red{o^b}\}$,and  $\OUT_8(o^q)$ is also $\{\red{o^a},\red{o^b}\}$. Then we focus on node 7: $\IN_7(o^p)=\OUT_5(o^p)\cup\OUT_8(o^p)=\{\red{o^a}\}\cup\{\red{o^a},\red{o^b}\}=\{\red{o^a},\red{o^b}\}$ and $ \IN_7(o^q)=\OUT_4(o^q)\cup\OUT_8(o^q)=\{\red{o^a},\red{o^b}\}\cup\{\red{o^a},\red{o^b}\}=\{\red{o^a},\red{o^b}\}$. Since node 7 performs a strong update for $o^q$, $\OUT_7(o^q)=\pts(a)=\{\red{o^a}\}$ while $\OUT_7(o^p)=\IN_7(o^p)=\{\red{o^a},\red{o^b}\}$. And the $\OUT_7$ is propagated to node 8, resulting in $\IN_8(o^p)=\{\red{o^a},\red{o^b}\}$ and $\IN_8(o^q)=\{\red{o^a}\}$. Restart analyzing node 8 again, the result is stable and reaches a fixed point.

\subsection{Incremental Pointer Analysis}
IPA is designed for flow-insensitive pointer analysis  which translates program statements into a set of constraints in Table~\ref{tab:andersen-constraints}, constructing a constraint graph (Definition \ref{def:cg}) \cite{andersen1994program}.  Flow-insensitive pointer analysis  iteratively propagates points-to sets along the constraint edges until a fixed point is reached. %

\begin{definition}[\textbf{Constraint Graph}]
	\label{def:cg}
The constraint graph is a directed graph $G = (\mathcal{N},\mathcal{E})$, where $\mathcal{N}\subseteq \mathcal{V}$ represents the set of variables  and $\mathcal{E}\subseteq \mathcal{N}\times\mathcal{N}$ represents subset constraints between point-to sets of connected variables.
\end{definition}

	\begin{table}
		\centering

		\captionof{table}{Constraints for flow-insensitive pointer analysis. The edge $p \leftarrow q \in \mathcal{E}$ means the constraint $\pts(q)\subseteq\pts(p)$. The edge $p \stackrel{f}{\leftarrow} q \in \mathcal{E}$ means the constraint $\forall o\in\pts(q),o.f\in\pts(p)$.}
		\label{tab:andersen-constraints}

		\begin{tabular}{c c c}
			\hline
			Type & Statement &  Points-To Constraint \\
			\hline
			Alloc & $p = alloc_o$  & $o \in \mathcal{P}(p)$ \\
			Copy & $p = q$ & $p \leftarrow q\in\mathcal{E}$  \\
			Phi & $p = \phi(p_1, \dots,p_n)$ & $\forall p_i : p \leftarrow p_i\in\mathcal{E}$  \\
			Load & $p = *q$ & $\forall o \in \mathcal{P}(q) : p \leftarrow o\in\mathcal{E}$ \\
			Store & $*p = q$ & $\forall o \in \mathcal{P}(p) : o\leftarrow q\in\mathcal{E}$ \\
			Field & $p = \&q.f$ & $p \xleftarrow{f} q\in\mathcal{E}$ \\
			\hline
		\end{tabular}
	\end{table}%

\textbf{Overall algorithm}.  The input to IPA is an acyclic constraint graph by  identifying SCCs and collapsing each SCC into a single representative node. %
Importantly, IPA leverage the \emph{change-locality property}~\cite{10.1145/3293606} of constraint graph: upon deleting a copy constraint edge, it suffices to check the incoming neighbors of the target node in the constraint graph to determine whether the potential removed objects should be removed and continue to propagate. %
For example, in Figure~\ref{fig:example_cg}, deleting $b\rightarrow o_2^x$ makes $\{\red{o^a},\red{o^b}\}$ candidate removals. Since the remaining predecessor $a$ still contains $\red{o^a}$, IPA preserves $\red{o^a}$ and propagates only the removal of $\red{o^b}$ from $\pts(c)$.

\subsection{Motivation}
\subsubsection{Transform the VFG to a Constraint Graph.}
\textbf{Identify the SCCs as IPA.} We aim to extend IPA's delta propagation~\cite{10.1145/3293606} to flow-sensitive pointer analysis. IPA must identify SCCs and use the predecessors of the entire SCC to filter removed objects; otherwise, imprecision may be introduced. In Figure \ref{fig:motivating_example} as an example, after deleting Line~\ref{line:strongupdate} (\texttt{ *x = b}), the points-to set of variable $c$ changes from $\{\red{o^a}, \red{o^b}\}$ to $\{\red{o^a}\}$(replaced by \texttt{a} from \texttt{b}), and the reduced set $\{\red{o^b}\}$ is propagated to its successor node 8. A local check would retain $\red{o^b}$ because it remains in $\OUT_7(o^p)$. However, since the occurrence of $\red{o^a}$ in both $\OUT_7(o^p)$ and $\OUT_8(o^p)$ is solely due to the propagation from $\pts(c)$, both $\OUT_8(o_p)$ and $\OUT_7(o^p)$ should decrease from $\{\red{o^a}, \red{o^b}\}$ to $\{\red{o^a}\}$. We cannot use $\OUT_7(o^p)$  to validate $\OUT_8(o^p)$ , as they are in the same SCC, where their values are interdependent and potentially changing together. %

\textbf{Strong updates may break SCCs presented on the VFG.}  Considering the example in Figure \ref{fig:motivating_example},  node 7 and node 8 form a single SCC.
Since node~7:$\texttt{*m = a}$ performs a strong update,  $\OUT_8(o^q)$ does not contribute to $\OUT_7(o^q)$. Therefore, the propagation path between $\OUT_7(o^q)$ and $\OUT_8(o^q)$ cannot form a SCC. If we mistakenly treat these two as belonging to the same SCC, when $\OUT_8(o^q)$ is reduced, the predecessors of $\OUT_7(o^q)$ may be incorrectly used to filter the potentially removed object, resulting in imprecision.

\begin{figure}[t]
	\centering
	\begin{minipage}[b]{0.35\linewidth}
		\includegraphics[width=\linewidth]{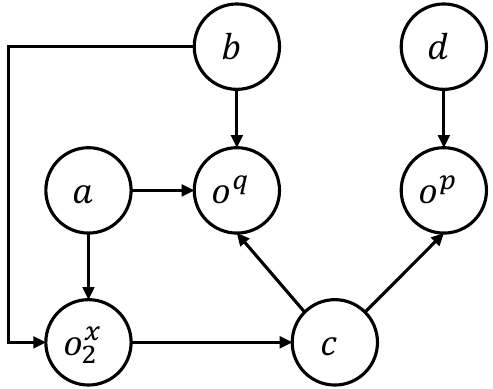}

		\caption{Example constraint graph in flow-insensitivity.}
		\label{fig:example_cg}
	\end{minipage}
\hspace{3mm}
	\begin{minipage}[b]{0.40\linewidth}
		\centering
		\includegraphics[width=\linewidth]{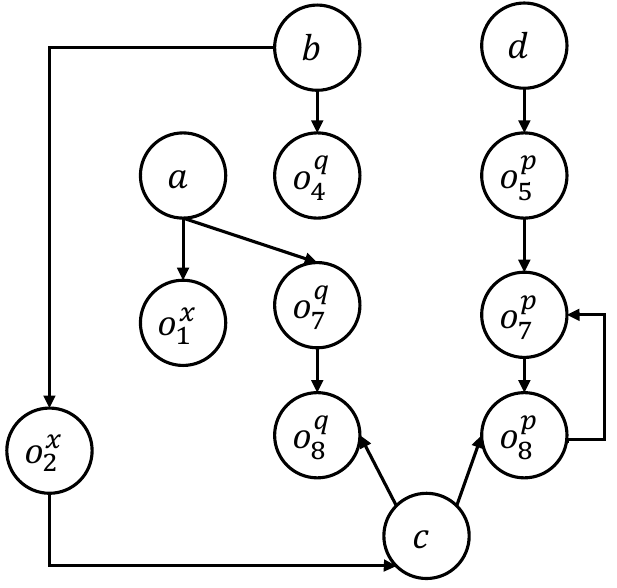}
		\caption{Constraint Graph with respect to the example VFG}
		\label{fig:constraint-graph}
	\end{minipage}
\end{figure}
 This highlights a critical issue: in flow-sensitive pointer analysis, strong updates for different address-taken variable will make these variables have a different dependency structure and the SCC decomposition varies per variable.
To efficiently reuse the \emph{change-local property} like IPA, our key insight is that SCC identification must be object-specific. As a result, we transform the statement-level value-flow graph into a variable-level constraint graph as shown in Figure \ref{fig:constraint-graph}. %
To preserve flow-sensitivity—,\emph{i.e.}, to distinguish the points-to set of a variable at different program statements— each potentially defined variable $v$ at statement $\ell$ is renamed as $v_\ell$, with $\pts(v_\ell)=\OUT_\ell(v)$. Specially, since a strong update happens at \texttt{*m = a},  $\OUT_8(o^q)$ does not contribute to the result of $\OUT_7(o^q)$ and there is no constraint edge from $o^q_8$ to $o^q_7$, resulting in no SCC for $o^q$. When the points-to set of \texttt{c} changes, the constraint graph explicitly encodes the propagation paths to $o^p$ and $o^q$ as well as the SCCs which they belongs to, enabling precise identification of which predecessor's old results can be safely reused.

\subsubsection{Propagation of Points-to Additions and Deletions.}

 Considering the example shown in Figure \ref{fig:example}, \texttt{*x = b} at Line ~\ref{line:strongupdate} performs a strong update and the points-to set of variable $c$ includes only  $b$'s points-to set $\{o^a, o^b\}$ . When Line~\ref{line:strongupdate} is removed, $\texttt{*x = a} \to \texttt{c = *x}$ becomes active and the points-to set of $c$ firstly is cleared and regains $\{o^a\}$. \red{IPA handles this incrementally by first propagating the reduced set ($\{o^a,o^b\}$) forward, and then propagating the increased set ($\{o^a\}$).} \red{This two-phase process introduces unnecessary overhead by repeatedly processing the same constraint nodes and performing SCC detection to break and reconstruct affected SCCs.}  %
 \red{For example, we extend the running example by incorporating the loop shown in Figure~\ref{fig:added-example}.} \red{This extension reuses variable $a$ and $c$ from Figure~\ref{fig:example}. The load at Line~11 induces the object-specific constraints $o^a_{11}\rightarrow e$ and $o^b_{11}\rightarrow e$. The store at Line~12, together with the loop-carried value flow to the next iteration's load, induces the return constraint $e\rightarrow o^a_{11}$. Hence, $\{o^a_{11},e\}$ forms an SCC. Under delete-then-insert propagation, the deletion phase temporarily changes $\pts(c)$ from $\{o^a,o^b\}$ to $\emptyset$ and removes $o^a_{11}\rightarrow e$, breaking the SCC and leaving only $e\rightarrow o^a_{11}$. The insertion phase then restores $o^a$, re-adds $o^a_{11}\rightarrow e$, and reconstructs the same SCC. Since SCC detection is linear in the graph size, repeatedly destroying and rebuilding such SCCs can dominate incremental-analysis time.}
 \begin{center}
 	\begin{minipage}{\columnwidth}
 		\begin{minipage}[c]{0.60\columnwidth}
 			\begin{lstlisting}[basicstyle=\ttfamily\footnotesize,numbers=left,firstnumber=10,numberstyle=\tiny,xleftmargin=1.8em]
while (...) {
	e = *c;
	*a = e;
}
 			\end{lstlisting}
 			\centering\red{\small (a) Code extension}
 		\end{minipage}\hfill
 		\begin{minipage}[c]{0.35\columnwidth}
 			\centering
 			\includegraphics[width=0.8\linewidth]{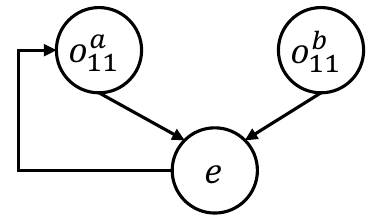}
 			
 			\red{\small (b) Local constraint graph}
 		\end{minipage}
 		\captionof{figure}{\red{Downstream extension of the running example and its relevant local constraint graph.}}
 		\label{fig:added-example}
 	\end{minipage}
 \end{center}

\red{To address these limitations, we propose a unified incremental algorithm that jointly handles both code deletions and insertions in a single pass. In the downstream-loop example, when the points-to set of \texttt{c} is potentially decreased by removing both $\red{o^a}$ and $\red{o^b}$, our algorithm uses not only the remaining edges but also the newly added edges to filter removed objects. Since the newly activated definition still contributes $\red{o^a}$,  only $\red{o^b}$ is removed and  edge $\red{o^a}_{11} \rightarrow e$ remains alive. Consequently, the SCC containing $\red{o^a}_{11}$ and $e$ is not destroyed, avoiding redundant SCC detection and reconstruction.}

\relax

\section{Incremental Full-Sparse Flow-Sensitive Pointer Analysis}

\subsection{Overall Incremental Algorithm}

\begin{algorithm}[t]
	\flushleft
	\caption{Incremental Analysis Algorithm for full-sparse flow-sensitive pointer analysis}
	\label{alg:getChangeMethod}
	\textbf{Input:} old points-to set $\mathcal{P}$, code changes $\mathcal{C}$, old VFG $G$ \\
	\textbf{Output:} flow pointer analysis result $\mathcal{P'}$ after modified\\
	\begin{algorithmic}[1]
		\STATE $G^\prime\leftarrow$ IncVFGbySILVA($G$,$\mathcal{C}$) \COMMENT{incrementally construct VFG by SILVA} \label{alg:start}
		\STATE $\mathcal{G}\leftarrow$TranVFGtoCG($G$) \COMMENT{transform VFG into constraint graph}
		\STATE $R,\mathcal{N}_T\leftarrow$SCCDetection($\mathcal{G}$)
		\COMMENT{perform SCC detection initially and merge SCCs}
		\STATE $N^+,N^-,E^+,E^-\leftarrow$ GraphDiff($G$,$G^\prime$)
		\STATE $\mathcal{E}^-,\mathcal{E}^+$=InitCGchanges($N^+,N^-,E^+,E^-$) \COMMENT{initialize deleted and inserted constraint edges}
		
		\WHILE{$\cg^-\neq \emptyset$ or $\cg^+\neq \emptyset$} \label{alg:mainstart}
		
		\FOR{$e_C:p\rightarrow q\in\cg^+$}
		\STATE P-InsCopy($e_C$,$R$) \COMMENT{handle inserted constraint edges}
		\ENDFOR
		
		\STATE$\mathcal{E}\leftarrow\mathcal{E}\cup\mathcal{E}^+\backslash \mathcal{E}^-$

		\STATE $R,\mathcal{N}_T\leftarrow$SCCDetection($\mathcal{G}$) \COMMENT{perform SCC detection and merge SCCs}
		\FOR{$e_C:p\rightarrow q\in\cg^-$}
		\STATE P-DelCopy($e_C$,$R$) \COMMENT{handle deleted constraint edges}
		\ENDFOR
		\STATE$\mathcal{E}^-,\mathcal{E}^+\leftarrow\emptyset$
		\STATE $\mathcal{P}^\prime\leftarrow $FilterAndPropagate($\mathcal{N}_T$, $\mathcal{P}$) \COMMENT{propagate and filter following topological order}
		\ENDWHILE\label{alg:mainend}
	\end{algorithmic}
\end{algorithm}

Our incremental full-sparse flow-sensitive pointer analysis algorithm proceeds in two phases. The first phase lies in Line \ref{alg:start}-\ref{alg:mainstart} in Algorithm.\ref{alg:getChangeMethod}.  The inputs of the algorithm are old flow- sensitive pointer analysis result $\pts$, old VFG $G$, and code changes $\mathcal{C}$. Firstly,  we  use the existing tool SILVA~\cite{10.1145/3725214} to incrementally compute the updated value-flow graph $G^\prime$ based on the code changes $\mathcal{C}$ and old VFG $G$.  Then we transform the old VFG $G$ into a constraint graph $\mathcal{G}$.\footnote{The detailed transformation rules are presented in \textbf{Section~\ref{sec:trans}}.} An SCC detection is performed on the constraint graph for subsequent use. Based on the detected SCCs, we collapse  each SCC into a single representative node on the constraint graph, yielding the representative information $R$ of each node, and obtain the topological order $\mathcal{N}_T$ .  Afterwards, we have implemented GraphDiff to compare the new value-flow graph $G^\prime$ with the old one $G$ and obtain the inserted nodes $N^+$, deleted nodes $N^-$, inserted value-flow edges $E^+$, deleted value-flow edges $E^-$.   \footnote{The graph changes in VFG will result in the deletion $\mathcal{E}^-$ and insertion $\mathcal{E}^+$ of constraint edges which is elaborated in \textbf{Section~\ref{sec:init}}.}

\par In the second phase which lies in Line \ref{alg:mainstart}-\ref{alg:mainend},  we firstly handle the inserted constraint edges $\mathcal{E}^+$  according to the representative information. We then perform SCC detection on the modified constraint graph and obtain the new representative information for the handling of the removed constraint edges $\mathcal{E}^-$. Following the topological order $\mathcal{N}_T$ derived from SCC detection, we compute and propagate the increases and decreases for each constraint node on the constraint graph $\mathcal{G}^\prime$. These changes in points-to set may further trigger  more copy constraint edges being deleted or inserted due to load  and store statements. We start the next iteration with these newly deleted edges $\mathcal{E}^-$ and inserted edges $\mathcal{E}^+$, and repeat this process until no further edge updates occur. \footnote{The precise mechanism for this efficient delta propagation is detailed in \textbf{Section~\ref{sec:filter}}.}

\subsection{Transform VFG to Constraint Graph}\label{sec:trans}

For each address-taken variable $o$ potentially defined at a store statement $\ell$, we create a renamed version $o_\ell$ with $\pts(o_\ell) = \OUT_\ell(o)$. After renaming, we build the constraint graph whose nodes  are either top level variables or renamed address taken variables. Inclusion based constraint edges between variables are added according to the semantic transfer function associated with each type of statements in flow-sensitive pointer analysis.

Figure \ref{fig:TransVFGToCG} illustrates the rules of transformation. For each \textbf{alloc} statement, we introduce a additional memory object $\pts(p)\leftarrow\pts(p)\cup\{o\}$ for p. For \textbf{copy} statements $p = q$ and \textbf{phi} instructions, we insert a copy constraint edge $p \leftarrow q$, enforcing the subset constraint $\pts(q) \subseteq \pts(p)$. For load and store statement $\ell$, since we only represent $\OUT_\ell(o)$ by $\pts(o_\ell)$ and do not present $\IN_\ell(o)$, we represent $\IN_\ell(o) $ by union the  $\pts(o_{\ell^\prime})$ is identical to $\OUT_{\ell^\prime}(o)$ if $\ell^\prime\stackrel{o}{\rightsquigarrow}\ell$. Consequently, for \textbf{load} statements ($p = *q$), for every object $o \in \pts(q)$ pointed to by the base variable $q$, we insert a copy constraint edge $ p \leftarrow o_{\ell'}$, effectively propagating the points-to information from the address-taken $o$ at the predecessor $\ell^\prime$ to variable $p$.

\begin{figure}[t]
	\centering
	\footnotesize

	\resizebox{\linewidth}{!}{
			$
	\begin{array}{c}
		\infer[\textsc{Alloc}]{\pts(p)\leftarrow\pts(p)\cup\{ o_\ell\}}{p=alloca_o}\vspace{1ex}\quad
		
		\infer[\textsc{Copy}]{\cg\leftarrow\cg\cup\{p\leftarrow q\}}{\ell:p=q}\vspace{1ex}\quad
		
		\infer[\textsc{Phi}]{\cg\leftarrow\cg \cup\{p\leftarrow p_i\}}{\ell:p=\phi(p_1,...p_n)}\vspace{1ex}\\
		
		\infer[\textsc{Field}]{\cg\leftarrow\cg\cup\{p\stackrel{f}{\leftarrow} q\}}{\ell:p=\&q.f}\vspace{1ex}\quad

		\infer[\textsc{Load}]{\cg\leftarrow\cg\cup \{p\leftarrow o_{\ell^\prime}\}}{\ell:p=*q,o\in pts(q),\ell^\prime \stackrel{o}{\rightsquigarrow} \ell}\vspace{1ex}\\
		
		\infer[\textsc{Store}]{\begin{cases}& |\pts(p)|=0\\
				\cg\leftarrow\cg\cup\{o_\ell \leftarrow q\},o^\prime\neq o\Rightarrow\cg\leftarrow\cg\cup\{o^\prime_\ell \leftarrow o^\prime_{\ell^\prime} \}&|\pts(p)|=1\wedge o\in \pts(p) \\
				\cg\leftarrow\cg\cup\{o^\prime_\ell\leftarrow o^\prime_{\ell^\prime} \}\cup\{o_\ell\leftarrow q\}	& |\pts(p)|>1\wedge o\in\pts(p)\\\end{cases}}{\ell:*p=q\quad \ell^\prime \stackrel{o^\prime}{\rightsquigarrow} \ell}\\
	\end{array}
$
}
	\caption{Rules of transforming VFG to constraint graph}
	\label{fig:TransVFGToCG}
\end{figure}
\par For each \textbf{store} statement ($\ell:*p = q$), we follow the standard SFS paradigm. To ensure correctness, we design three specific constraint generation rules to align with these distinct semantic transfer functions:
\begin{itemize}[leftmargin=12pt,itemsep=2pt,topsep=2pt]
	\item  $\vert\pts(p)\vert = 0$: we  does not insert any constraint edges.
	
	\item  $\vert\pts(p)\vert = 1 \wedge o\in\pts(p)$: we insert a constraint edge $o_\ell \leftarrow q$, propagating the points-to set of $q$ to $o_\ell$. For all other variable, $o'$  (where $o' \neq o$), we preserve their values propagated from predecessors by adding copy constraint edges $o'_\ell \leftarrow o'_{\ell'} $.
	
	\item  $\vert\pts(p)\vert > 1\wedge o\in\pts(p)$: we insert a copy constraint $ o_\ell \leftarrow q$ and a copy constraint edge $ o^\prime_\ell \leftarrow o^\prime_{\ell'}$ for all potential defined $o^\prime$ to  preserve points-to information from the predecessor.
\end{itemize}

\subsection{Initialize by Deleting and Inserting Edges According to VFG Changes}\label{sec:init}
\begin{figure}[t]
	\centering
	\footnotesize
	\[
	\begin{array}{c}
		\infer[\textsc{DelAlloc}]{\Delta^-_p\leftarrow\Delta^-_p\cup\{o\}}{\ell:p=alloc_o\in N^-}\vspace{1ex} \quad
		
		\infer[\textsc{DelCopy}]{\cg^-\leftarrow\cg^-\cup\{p\leftarrow q\}}{\ell:p=q\in N^-}\vspace{1ex}\\
		
		\infer[\textsc{DelField}]{\cg^-\leftarrow\cg^-\cup\{p\stackrel{f}{\leftarrow} q\}}{\ell:p=\&q.f\in N^-}\vspace{1ex} \quad
		
		\infer[\textsc{DelPhi}]{\cg^-\leftarrow\cg^-\cup\{p\leftarrow p_i\}}{\ell:p=\phi(p_1,...p_n),i\in[1,n]\in N^-}\vspace{1ex}\\
		
		\infer[\textsc{DelLoad}]{\cg^-\leftarrow \cg^-\cup \{p\leftarrow o_{\ell^\prime}\}}{\ell:p=*q\in G,o\in \pts(p),\ell^\prime\stackrel{o}{\rightsquigarrow}\ell\in E^-}\vspace{1ex}\\
		
		\infer[\textsc{DelStore}]{\shortstack[l]{\(\ell\in N^-\Rightarrow \cg^-\leftarrow \cg^-\cup\{o_{\ell}\leftarrow q\ |\ o\in\pts(p)\}\)\\
				\(\begin{cases} & |\pts(p)| = 0\\ \cg^-\leftarrow \cg^-\cup\{o_\ell \leftarrow o_{\ell^\prime}\} & |\pts(p)| = 1\wedge o\notin\pts(p)\\ \cg^-\leftarrow \cg^-\cup\{o_\ell \leftarrow o_{\ell^\prime}\} & |\pts(p)| > 1\\ \end{cases}\)} }{\ell:*p=q\in G,\ell^\prime\stackrel{o}{\rightsquigarrow}\ell\in E^-}\vspace{1ex}\\
	\end{array}
	\]
	\caption{Remove VFG nodes and edges}
	\label{fig:deledges}
\end{figure}

In this section, we initilize the removed and inserted constraint edges  based on the VFG changes: $E^+, E^-,N^-, N^+$ and initialize both positive deltas $\Delta^+$ and negative deltas $\Delta^-$ in the points-to set according to the constraint graph changes.

\par \textbf{Delete VFG nodes and edges}: as illustrated in Figure \ref{fig:deledges},  
For address, load, copy, phi, and field statements,the deletion process  suffices to remove all incoming constraint edges whose target is the  defined top-level variable 
For instance, deleting a copy statement $p=q$ simply involves removing the edge $p \leftarrow q$.

For \textbf{store} nodes ($\ell: *p = q$), when this node is deleted,\emph{i.e.},$\ell\in N^-$, the constraint edge $o_\ell\leftarrow q$ is removed if p points-to $o$. Afterwards, we remove the constraint edge which propagates points-to set of predecessors to this node based on $p$'s points-to set:

\begin{itemize}[leftmargin=12pt,itemsep=2pt,topsep=2pt]
	\item  $|\mathcal{P}(p)| = 0$:  there is no constraint edge propagating the points-to set from the predecessor. Upon deletion, we do nothing. 
	
	\item  $|\mathcal{P}(p)| = 1 \wedge o\in\pts(p)$: for all $o\notin \pts(p)$, since there exists a constraint edge $o_\ell\leftarrow o_{\ell^\prime}$, upon the value flow edge $\ell^\prime\stackrel{o}{\rightsquigarrow}\ell$ is deleted, we remove $o_\ell \leftarrow o_{\ell'} $.
	
	\item $|\mathcal{P}(p)| > 1 \wedge o\in\pts(p)$: we remove the constraint edges $ o_\ell \leftarrow o_{\ell^\prime}$ upon the value flow edge $\ell^\prime\stackrel{o}{\rightsquigarrow}\ell$ is deleted.
\end{itemize}

\par \textbf{Insert VFG nodes and edges.} As illustrated in Figure \ref{fig:insedges}, for insertion of \textbf{load}  statement ($\ell: *p = q$) and its inserted value flow edge $\ell^\prime\stackrel{o}{\rightsquigarrow} \ell$, we insert the constraint edge $o_{\ell^\prime}\leftarrow p$. For insertion of \textbf{store} statement ($\ell: *p = q$), we firstly insert the newly generated constraint edge $o_\ell\leftarrow q$ if p points-to $o$. Afterwards,we insert the constraint edge if a value flow edge $\ell^\prime\stackrel{o}{\rightsquigarrow} \ell$ is inserted to this node based on $p$'s points-to set:

\begin{itemize}[leftmargin=12pt,itemsep=2pt,topsep=2pt]
	\item $|\mathcal{P}(p)| = 0$: no constraint edges are inserted.
	
	\item $|\mathcal{P}(p)| = 1 \wedge o\in\pts(p)$:  upon a value flow edge $\ell^\prime\stackrel{o}{\rightsquigarrow}\ell$ in inserted,  constraint edge $o_\ell \leftarrow o_{\ell'}$  is added for $o\notin\pts(p)$.
	
	\item $|\mathcal{P}(p)| > 1\wedge o\in\pts(p)$: upon a value flow edge $\ell^\prime\stackrel{o}{\rightsquigarrow}\ell$ is inserted, we  insert a constraint edge $o_\ell \leftarrow o_{\ell'}$   to preserve the incoming flow.
\end{itemize}
\begin{figure}[t]
	\centering
	\footnotesize
	\[
	\begin{array}{c}
		\infer[\textsc{InsAlloc}]{\Delta^+_p\leftarrow\Delta^+_p\cup\{o\}}{\ell:p=alloc_o\in N^+}\vspace{1ex} \quad	
		\infer[\textsc{InsCopy}]{\cg^+\leftarrow\cg^+\cup\{p\leftarrow q\}}{\ell:p=q\in N^+}\vspace{1ex}\\
		
		\infer[\textsc{InsPhi}]{\cg^+\leftarrow\cg^+\cup\{p\leftarrow p_i\}}{\ell:p=\phi(p_1,...p_n),i\in[1,n]\in N^+}\vspace{1ex}\quad
		
		\infer[\textsc{InsField}]{\cg^+\leftarrow\cg^+\cup\{p\stackrel{f}{\leftarrow} q\}}{\ell:p=q\in N^+}\vspace{1ex} \\

		\infer[\textsc{InsLoad}]{\cg^+\leftarrow\cg^+\cup \{p\leftarrow o_{\ell^\prime}\}}{\ell:p=*q\in G^\prime\quad o\in \pts(q)\quad \ell^\prime \stackrel{o}{\rightsquigarrow }\ell \in E^+}\vspace{1ex}\\
		
		\infer[\textsc{InsStore}]{\shortstack[l]{\(\ell\in N^+\Rightarrow \cg^+\leftarrow \cg^+\cup\{o_{\ell}\leftarrow q\ | \ o\in\pts(p)\}\)\\
				\(\begin{cases} & |\pts(p)| = 0\\ \cg^+\leftarrow \cg^+\cup\{o_\ell \leftarrow o_{\ell^\prime}\} & |\pts(p)| = 1\wedge o\notin\pts(p)\\ \cg^+\leftarrow \cg^+\cup\{o_\ell \leftarrow o_{\ell^\prime}\} & |\pts(p)| > 1\\ \end{cases}\)} }{\ell:*p=q\in G^\prime,\ell^\prime\stackrel{o}{\rightsquigarrow}\ell\in E^+}\vspace{1ex}\\
	\end{array}
	\]
	\caption{Insert VFG nodes and edges}
	\label{fig:insedges}
\end{figure}
\vspace{-3mm}
\begin{example}
Considering the example in Fig.~\ref{fig:motivating_example}, 
upon deleting node 2, we modify the value-flow graph by removing value flow edges $1 \stackrel{o^x}{\rightsquigarrow} 2$ and $2 \stackrel{o^x}{\rightsquigarrow} 3$, and inserting a new value flow edge $1 \stackrel{o^x}{\rightsquigarrow} 3$. 
Correspondingly, we update the constraint graph: 
(1) The direct constraint $o^x_2 \leftarrow b$ is removed  according to the Rule \textsc{DelStore}. 
(2) For the deleted incoming value flow $1 \stackrel{o^x}{\rightsquigarrow} 2$, since the strong update at node 2 previously killed the definition from Node 1, no predecessor constraint $o^x_2 \leftarrow o^x_1$ existed to be removed according to the Rule \textsc{DelStore}. 
(3) For the deleted outgoing flow $2 \stackrel{o^x}{\rightsquigarrow} 3$, we remove the constraint $c \leftarrow o^x_2$ following the \textsc{DelLoad} rule. 
(4) Conversely, for the newly inserted flow $1 \stackrel{o^x}{\rightsquigarrow} 3$, we add the constraint $c \leftarrow o^x_1$ according to the \textsc{InsLoad} rule.
\end{example}

\subsection{Filter and Propagate Difference of Points-to Set on Constraint Graph}\label{sec:filter}

\par  After identifying the structural changes in the constraint graph, we first initialize the potential changes in variables’ points-to sets based on the insertion and deletion of copy edges (Rules \textbf{\textsc{P-InsCopy}} and \textbf{\textsc{P-DelCopy}}).
Given these initial deltas, we apply the \textbf{\textsc{Filter}} rule to eliminate spurious points-to updates using the \emph{change-local property} of the constraint graph and the \textbf{\textsc{Propagate}} rule to propagate the actual changes along the graph.
Such propagation may further trigger deletion or insertion of constraint edges (\textbf{\textsc{P-Store}} and \textbf{\textsc{P-Load}}). After each propagation round, the newly inserted or deleted constraint edges are treated as the input for the next iteration until a fixed point is reached.  The inference rules governing this process are presented in Figure \ref{fig:filterAndProp}. 

\textbf{Initialization of Deltas in Points-to sets}.
When a copy constraint edge $p \leftarrow q$ is inserted, if p and q previously share the same representative node,\emph{i.e.},$R(p) = R(q)$ which means the points-to sets of both p and q are the same, thus q does not introduce any additional objects. Otherwise, we initialize the positive delta of $p$ as $\Delta^+_p \leftarrow \Delta^+_p \cup \mathcal{P}(q)$, indicating that $p$ may now point to new objects in $q$'s points-to set. We then process the deleted constraint edges $p \leftarrow q$. If p and q are previously in the same SCC, then this SCC may be broken. We perform an SCC detection on the updated constraint graph. If p and q do not belong to the same SCC in the updated graph, we initialize the negative delta as $\Delta^-_p \leftarrow \Delta^-_p \cup \pts(q)$, indicating that $p$ may lose the objects in $q$'s points-to set.
\begin{figure}[t]
	\centering
		{\footnotesize
		\[\begin{array}{c}
			
			\infer[\textsc{P-InsCopy}]{
				R(p)\neq R(q)\Rightarrow \Delta^+_p\leftarrow\Delta^+_p\cup\pts(q)}{q\rightarrow p\in\cg^+}\vspace{1ex}\\\
			
			\infer[\textsc{P-DelCopy}]{
				R(p)\neq R(q)\Rightarrow\Delta^-_p\leftarrow\Delta^-_p\cup\pts(q)}{q\rightarrow  p\in\cg^-} \\

			\infer[\textsc{Filter}]{
				\shortstack[l]{
					\(\forall p \leftarrow q\in\mathcal{E}\Rightarrow \Delta^-_p \leftarrow \Delta^-_p \setminus \pts(q)\)\\
					\(\Delta^+_p \leftarrow \Delta^+_p \setminus \pts(p)\) \quad
					\(\Delta^+_p \leftarrow \Delta^+_p \setminus \Delta^-_p\) \quad 
					\(\Delta^-_p\leftarrow\Delta^-_p\cap \pts(p)\) \\
					\(\pts(p)\leftarrow\pts(p)\backslash\Delta^-_p\cup\Delta^+_p\quad \textsc{Propagate}(p,\Delta^-_p,\Delta^+_p)\)}
			}{p, \Delta^-_p, \Delta^+_p}
			\vspace{1em}\\
			
			\infer[\textsc{Propagate}]{\shortstack[l]{\(\forall  s\leftarrow p\in\mathcal{E}, \Delta_s^+\leftarrow \Delta_s^+\cup\Delta_p^+, \Delta_s^-\leftarrow\Delta_s^-\cup\Delta_p^-\)\\
					\(\forall \ell:*p=q\in G^\prime, \ \textsc{P-STORE}(\ell:*p=q, \Delta^-_p, \Delta^+_p)\) \\
					\(\forall \ell:q=*p\in G^\prime, \ \textsc{P-LOAD}(\ell:q=*p, \Delta^-_p, \Delta^+_p)\)}
			}{p, \Delta^-_p, \Delta^+_p}
			\vspace{0.5em}\\
			
			\infer[\textsc{P-Load}]
			{
				\cg^-\leftarrow\cg^-\cup\{p\leftarrow \indelo\}\quad
				\cg^+ \leftarrow\cg^+\cup\{p\leftarrow \ininso \}
			}
			{
				\ell:p=*q,o^+\in\Delta_p^+,o^-\in\Delta_p^+,\ell^\prime\stackrel{o^+}{\rightsquigarrow}\ell,\ell^\prime\stackrel{o^-}{\rightsquigarrow}\ell
			}\vspace{1em}\\
			
			\infer[\textsc{P-Store}]
			{
				\shortstack[l]{
					\(|\pts(p)|=0 \Rightarrow \mathcal{E^-}\leftarrow\mathcal{E^-}\cup\{e_C|e_C:o_\ell \leftarrow o_{\ell^\prime} \in \mathcal{E}\}\)\\
					\(|\pts(p)|=1\Rightarrow\begin{cases}
						\mathcal{E^-}\leftarrow\mathcal{E^-}\cup\{e_C|e_C: o_\ell\leftarrow o_{\ell^\prime}  \in \mathcal{E}\} & o\in\pts(p)\\
						\mathcal{E^+}\leftarrow\mathcal{E^+}\cup\{e_C|e_C:o_\ell\leftarrow o_{\ell^\prime}\notin \mathcal{E}\} & o\notin \pts(p)\\
					\end{cases} \)\\
					\(|\pts(p)|>1 \Rightarrow \mathcal{E}^+\leftarrow\mathcal{E}^+\cup\{e_C|e_C:o_\ell\leftarrow o_{\ell^\prime} \notin \mathcal{E}\}\)\\
					\(\mathcal{E}^-\leftarrow\mathcal{E}^-\cup\{o^-_\ell\leftarrow q \},\mathcal{E}^+\leftarrow\mathcal{E}^+\cup\{ o^+_\ell\leftarrow q\}\)
				}
			}
			{ \ell:*p=q,\ell^\prime\stackrel{o}{\rightsquigarrow}\ell,,\;o^+\in\Delta_p^+,\;o^-\in\Delta_p^+,\;
			}\vspace{1ex}\\

		\end{array}
		\]
			}
	\caption{Rules of filtering and propagating increase and decrease in the points-to set }
	\label{fig:filterAndProp}
\end{figure}

\par \textbf{Topological Propagation and Filtering}.  We process constraint nodes in topological order to handle dependencies correctly. The initial deltas $\Delta^+$ and $\Delta^-$ derived from edge insertions/deletions or propagated from predecessors represent only \emph{candidate} changes.  \begin{itemize}[leftmargin=12pt,itemsep=2pt,topsep=2pt]
	\item \textbf{Filtering Decreases ($\Delta^-_p$):}  for each object $o$ in the candidate $\Delta^-_p$, we inspect all existing predecessors $\{q \mid p\leftarrow q \in \mathcal{G}\}$. If there exists any predecessor $q$ such that $o \in \pts(q)$,  we remove $o$ from $\Delta^-_p$ and $\pts(p)$ retains $o$. 
	
	\item \textbf{Filtering Increases ($\Delta^+_p$):}  if a candidate increased object $o$ is already present in $\pts(p)$, it constitutes a redundant update. Thus, we remove those objects from $\Delta^+_p$ which is already in $\pts(p)$.
\end{itemize}

\begin{figure}
	\includegraphics[width=0.6\linewidth]{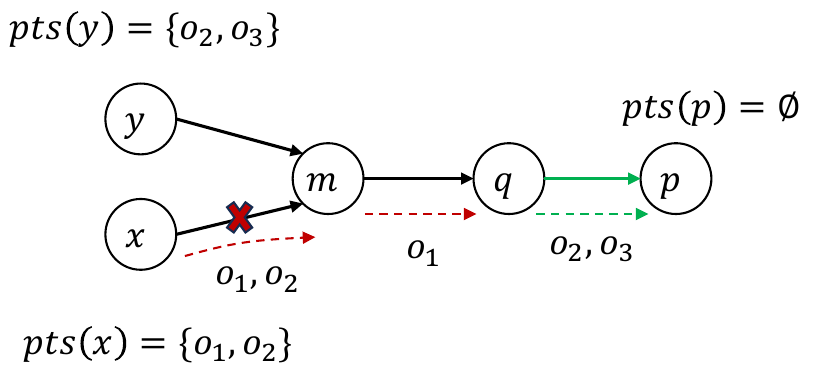}
	\caption{Simultaneous propagation of points-to additions and deletions.}
	\label{fig:topo}
\end{figure}A critical challenge in incremental analysis is handling "interleaved" updates, where a variable might receive both additions and deletions in the same iteration. 
For instance, as shown in Figure \ref{fig:topo}, if a new edge $q \to p$ is inserted while $q$ itself is losing a target $o_1$ due to an upstream deletion $x \to m$, eagerly adding $\pts(q)=\{o_1,o_2,o_3\}$ to $\Delta^+_p$ would incorrectly propagate $o_1$ to $p$.
To address this, the \textbf{\textsc{Filter}} rule applies strict set difference operations: we remove any element found in $\Delta^-_p$ from $\Delta^+_p$ ,\emph{i.e.}, $\Delta^+_p \leftarrow \Delta^+_p \setminus \Delta^-_p$). This also guarantees that $\Delta^+_p$ and $\Delta^-_p$ are mutually exclusive.  We subtract $\Delta^-_p=\{o_1\}$ from $\Delta^+_p=\{o_1,o_2,o_3\}$, resulting in the updated $\Delta^+_p=\{o_2,o_3\}$. The \textbf{\textsc{Filter}} rule also constrains the negative delta to include only elements that are actually present in the current points-to set ,\emph{i.e.}, $\Delta^-_p \leftarrow \Delta^-_p \cap \pts(p)$). In this example, assuming $p$ is initially empty or does not contain $o_1$, the operation yields $\Delta^-_p = \{o_1\} \cap \pts(p) = \emptyset$. %

\textbf{Handling Side Effects on Load/Store Statements.}
Finally, we update the points-to set of $p$ locally: $\mathcal{P}(p) \leftarrow (\mathcal{P}(p) \setminus \Delta^-_p) \cup \Delta^+_p$, and propagate the refined $\Delta^+_p$ and $\Delta^-_p$ to $p$'s successors. Changes in the points-to set of a pointer variable $p$ ($\Delta_p^+$ or $\Delta_p^-$) can trigger more constraint edge updates.

Load Statements ($\ell: p = *q$).
As defined in rule \textbf{\textsc{P-Load}},  for any object $o$ added to $q$'s points-to set,\emph{e.g.},$o \in \Delta_q^+$), we insert a copy constraint edge $p \leftarrow o_{\ell'}$. Conversely,constraint edge $p \leftarrow o_{\ell'}$ is removed for object $o \in \Delta_q^-$.

Store Statements ($\ell: *p = q$).
As defined in rule \textbf{\textsc{P-Store}},  we add a constraint edge $ o^+_\ell\leftarrow q$ and remove $ o^-_\ell\leftarrow q$ for $o^+ \in \Delta_p^+$ and $o^- \in \Delta_p^-$ respectively.
Additionally, we also incrementally modify whether the points-to information of current variable should be propagated from the predecessor depending on $p$'s points-to set:
\begin{itemize}[leftmargin=12pt,itemsep=2pt,topsep=2pt]
	\item $|\mathcal{P}(p)| = 0$: the store becomes a no-op and any existing bypass edges $o_\ell \leftarrow o_{\ell'}$ is deleted.
	\item $|\mathcal{P}(p)| = 1$: for the unique target object $o \in \mathcal{P}(p)$, we delete the constraint edge $o_\ell \leftarrow o_{\ell^\prime}$ to kill the points-to information propagated from predecessor $\ell^\prime$ if it exists . For all other objects $o' \notin \mathcal{P}(p)$, the bypass edge $o'_\ell \leftarrow o'_{\ell'} $ is inserted if missing
	.
	\item $|\mathcal{P}(p)| > 1$: For any object $o$ that reaches this statement, we must insert the bypass edge $ o_\ell \leftarrow o_{\ell'}$ if it is missing, restoring the points-to information propagated from the predecessor.
\end{itemize}

\begin{lemma}[Monotonicity]
	\label{lemma:mon}
	\red{Points-to reductions and expansions induce only edge deletions and insertions, respectively, which in turn can only shrink and enlarge points-to sets.
	}
\end{lemma}

\begin{proof}
	\red{Consider a store statement $\ell: *p=q$. whose generated constraint edges depend on $\pts(p)$. If $\pts(p)$ decreases from multiple objects to a singleton ${o}$, the update changes from weak to strong: all bypass edges for $o$ and all edges from $q$ to objects removed from $\pts(p)$ are deleted, while bypass edges for the remaining objects are preserved. If the store becomes a no-op, all associated constraint edges are removed. Therefore, the resulting target edge set can only decrease. Conversely, if $\pts(p)$ expands, causing a transition from a no-op to a strong or weak update, or from a strong update to a weak update, the target edge set can only increase. Edge deletions only trigger $\Delta^-$ and further edge deletions, whereas edge insertions only enlarge them and trigger further insertions.
	}

\end{proof}

\begin{definition}[Object points-to cycle]
	\label{def:object-pts-cycle}
	\red{For object nodes $o,o'$, write $o\leadsto_{\mathrm{pt}}o'$ iff $o'\in\pts(o)$. An object points-to cycle is a nonempty sequence $o_1,\ldots,o_n$ such that $o_i\leadsto_{\mathrm{pt}}o_{i+1}$ for $1\le i<n$ and $o_n\leadsto_{\mathrm{pt}}o_1$. The points-to relation is \emph{object-acyclic} iff it contains no such cycle.}
\end{definition}

\begin{theorem}
	\red{If the points-to relation remains object-acyclic during the analysis, the interleaved analysis algorithm terminates.}
\end{theorem}

\begin{proof}
	\red{Let $S_k=\langle G_k,\mathcal{E}_k^-,\mathcal{E}_k^+\rangle$ be the state after round $k$, where $G_k$ is the constraint graph annotated with the current points-to sets and $\mathcal{E}_k^-$ and $\mathcal{E}_k^+$ are the load/store-derived edge deletions and insertions.  If the algorithm did not terminate, then $S_i=S_j$ for some $i<j$. Since filtering makes $\mathcal{E}_m^-\cap\mathcal{E}_m^+=\emptyset$, every edge deleted between these two equal states must be reinserted in a later round, and conversely. Thus,
	$
	D_{i,j}^-=\bigcup_{i\le m<j}\mathcal{E}_m^-
	=\bigcup_{i\le m<j}\mathcal{E}_m^+=D_{i,j}^+\neq\emptyset.
	$
	Let $e\Rightarrow e'$ mean that the points-to reduction caused by derived-edge deletion $e$ triggers deletion $e'$ through \textsc{P-Load} or \textsc{P-Store} according to Lemma~\ref{lemma:mon}. Since every update in $D_{i,j}^-$ must be reversed before $S_j$, following these causes yields a cycle $e_1\Rightarrow e_2\Rightarrow\cdots\Rightarrow e_n\Rightarrow e_1$. Let $e_r(1\leq r\leq n)$ be induced when the store $*p_r=a_r$ loses target object $o_r$. If $e_r$ causes $e_{r+1}$, the reduction propagated through $o_r$ removes $o_{r+1}$ at the next store; thus $o_{r+1}\in\pts(o_r)$, with indices modulo $n$. Hence,
	$
	o_2\in\pts(o_1),\ o_3\in\pts(o_2),\ldots,\ o_1\in\pts(o_n),
	$
	which gives $o_1\leadsto_{\mathrm{pt}}o_2\leadsto_{\mathrm{pt}}\cdots\leadsto_{\mathrm{pt}}o_n\leadsto_{\mathrm{pt}}o_1$. %
	This contradicts the acyclicity assumption. Therefore, the algorithm terminates.}
	
\end{proof}

\begin{theorem}[Consistency]
	\red{Under the object-acyclicity condition , for each statement type in the updated program, the constraint graph and points-to sets produced by our incremental method are consistent with those derived from the data-flow functions of full-sparse flow-sensitive pointer analysis.}
\end{theorem}

\begin{proof}[Proof Sketch]
	\red{Consider a copy statement $\ell:p=q$; the same argument applies to field and phi statements. We prove that $\pts(q)\subseteq\pts(p)$. Assume, for contradiction, that some object $o\in\pts(q)$ but $o\notin\pts(p)$. If $o\in\Delta^-_p$, the deletion-filtering rule $\Delta^-_p\leftarrow\Delta^-_p\setminus\pts(q)$ removes $o$ from $\Delta^-_p$. Otherwise, $o$ must be absent from $\Delta^+_p$. However, an object already in $\Delta^+_p$ can be filtered out only if it is already in $\pts(p)$. Another case where $q\rightarrow p$ is newly inserted also results in $o\in\Delta^+_p$. Thus, the assumption is contradicted and $\pts(q)\subseteq\pts(p)$. For Store and load statements, rule \textbf{P-Store} and \textbf{P-Load} guarantee that constraint graph are consistent with those derived from data-flow functions. }

	\red{We prove the minimality under object-acyclic condition. Assume the resulting fixed point is non-minimal. That is if $o_1$ is to be removed from $\pts(p)$,  $r\rightarrow p$ filters the removal of $o_1$ regardless of whether this edge is pre-existing or newly inserted during  interleaved propagation. If removing $o$ must eventually delete $r\rightarrow p$. There must exist a chain $e_1\Rightarrow e_2\Rightarrow\cdots\Rightarrow e_n$, where $e_n=(r\rightarrow p)$. Each step is induced by a store $*p_k=a_k$: losing target $o_k$ deletes $a_k\rightarrow o_k$ and causes $o_k$ to lose $o_{k+1}$, which implies
	$
	o_1\in\pts(p), o_2\in\pts(o_1),\ldots, p\in\pts(o_{n-1}),
	$
	which contradicts object-acyclicity. Hence interleaved analysis will reach the least fixed point.}
\end{proof}

\red{Complete proofs of termination, constraint consistency, and least-fixed-point convergence are provided in the supplementary material archived with the artifact (Section~\ref{sec:data-availability}).}

\section{Evaluation}
\subsection{Implementation}
 We implemented our approach, named \textbf{IncSFS}, on top of the \textsc{SVF} framework (v2.7)~\cite{10.1145/2892208.2892235}, utilizing LLVM 14.0.1 as the underlying infrastructure. \textsc{SVF} is a widely-used, actively maintained open-source framework that supports full-sparse flow-sensitive pointer analysis. To handle the initial incremental updates of the Value-Flow Graph (VFG), we integrated the \textsc{Silva}~\cite{10.1145/3725214} framework into IncSFS. %

We designed our experiments to answer the following questions:

\begin{itemize}[leftmargin=12pt,itemsep=2pt,topsep=2pt]
		
	\item \textbf{RQ1 (Performance Gain):} How much speedup does \textbf{IncSFS} achieve compared to whole-program analysis and traditional \emph{Reset–Recompute} strategy? 
	\item \textbf{RQ2 (Advantage over Two-Phase IPA):} How much performance benefit does IncSFS gain from its \textbf{unified propagation} strategy on the \textbf{constraint graph} compared to the traditional delete-then-insert propagation?
	\item \textbf{RQ3 (Scalability):} How does the analysis time of \textbf{IncSFS} scale with an increasing number of commits?
	\item \textbf{RQ4 (Correctness):} Are the points-to results produced incrementally by \textbf{IncSFS} identical to those obtained from the whole-program analysis? \end{itemize}

\vspace{-5mm}
\subsection{Benchmark and Experimental Setup}
	
	\begin{table} %
		\centering
		\footnotesize
\captionof{table}{Subject programs and its characteristics. Prog, Loc,Nodes,Edges,Pointers,Objects,Funs,CS represents programs, the number of lines of code, nodes in VFG, number of  edges in VFG, number of pointer variables , number of memory objects, number of functions, number of callsite.} 
		\label{tab:benchmarks}
		
		\resizebox{\linewidth}{!}{
			\begin{tabular}{c c c c c c c c}
				\toprule
				Prog & Loc & Nodes & Edges & Pointer & Object & Funs & CS \\
				\midrule
				janet  & 85274  & 134420 & 200093 & 176155 & 4300  & 1544 & 9676 \\
				zstd   & 95939  & 380797 & 558708 & 582876 & 49616 & 2348 & 24627 \\
				tmux   & 74322  & 231969 & 374530 & 207267 & 6215  & 2228 & 16524 \\
				astyle & 100424 & 281375 & 467675 & 321716 & 4845  & 712  & 20658 \\
				nginx  & 170660 & 171730 & 284550 & 170058 & 2568  & 1424 & 9206 \\
				sqlite & 525332 & 457787 & 756968 & 520887 & 8471  & 2566 & 28304 \\
				\bottomrule
			\end{tabular}
		}
		\footnotesize
		\raggedright 
	\end{table}%

\textbf{Subject Programs.} We evaluated IncSFS on a suite of six large-scale, real-world programs as shown in table~\ref{tab:benchmarks} characterized by frequent updates, which were selected based on their widespread adoption in prior flow-sensitive pointer analysis literature~\cite{9370334,10.1145/3725214}. 

\textbf{Dataset Construction.} To assess the impact of code change magnitude, we constructed datasets based on four distinct commit intervals: 1, 10, 20, and 30. An interval of $N$ represents the evolution between a baseline commit and a version $N$ commits forward. For each program and interval, we randomly sampled 10 distinct pairs of old and new versions from the repository history. \red{For every interval, we exclude  pairs with no source-code differences between the old and new versions in the process of random selection. } The reported performance for each interval is derived from the average analysis time of these 10 sampled pairs. Upon manual inspection, 81.67\% of our randomly sampled pairs involve modifications such as product releases, bug fixes, and code refactoring. 

\red{\textbf{Source-to-IR Mapping and IR Noise.} Since minor source edits may cause noisy whole-IR differences, Our current incremental approach have
	to also handle IR deletions followed by insertions. we map source changes directly to the corresponding deleted IR fragment $P_d$ in the old IR $I_o$ and inserted IR fragment $P_i$ in the new IR $I_n$. We measure the deletion time $T_d$ by removing $P_d$ from $I_o$. For insertion, we obtain an equivalent intermediate state by fully analyzing $I_n$, removing $P_i$, and then reinserting it to measure $T_i$. The reported incremental time is $T_d+T_i$, including only the transformation from the VFG to the constraint graph, SCC detection, and incremental propagation. The reported time of incremental analysis excludes the parts of LLVM IR generation, source-to-IR mapping, and incremental VFG maintenance which is a common prerequisite shared by all incremental algorithms evaluated.}

\textbf{Baselines.} \red{We
	implemented three incremental strategies: 
	Reset-Recompute, the adapted SILVA two-phase delete-then-insert
	propagation, and IncSFS's unified single-phase propagation.  To answer \textbf{RQ1}, we compare it against  the classical Whole-Program Analysis (Full)~\cite{5764696} since it is the standard implementation provided by SVF; the Reset-Recompute strategy, which reset and recomputes all nodes reachable from the changed VFG nodes/edges.} We explicitly exclude the \emph{Restart-Iteration} strategy from our evaluation since it inherently suffers from precision loss.  To answer \textbf{RQ2}, we compare IncSFS with SILVA~\cite{10.1145/3725214}, the state-of-the-art flow-insensitive incremental pointer analysis. For both approaches, we measure the runtime of delta computation and propagation on the constraint graph transformed from the VFG. \red{In this graph, all top-level and address-taken variables are converted into Full SSA form, allowing flow-sensitive pointer analysis to be solved flow-insensitively and SILVA to be directly applied.}
 To answer \textbf{RQ3}, \red{we evaluate representative small (\texttt{janet}), medium (\texttt{tmux}), and large (\texttt{sqlite}) programs against their initial versions, increasing the commit interval by 100 for \texttt{janet} and \texttt{tmux} and by 500 for \texttt{sqlite}, until incremental analysis approaches or exceeds full re-analysis.} To answer \textbf{RQ4}, we compare the points-to sets of all maintained top-level variable produced by IncSFS against those computed by the  full analysis across all experimental configurations. We also compute the average size of point-to set between IncSFS and the  full analysis.

\noindent\textbf{Experimental Environment.} All experiments were conducted on a high-performance workstation running Ubuntu 20.04 LTS. The workstation features an AMD Ryzen Threadripper PRO 7995WX processor with 128 CPU cores operating at 2.5\,GHz and 512\,GB of physical memory.
\vspace{-3mm}
\subsection{Experimental Results}
\subsubsection{Comparison with Full Analysis and Reset-Recompute}
\begin{table}[t]
	\scriptsize
	\caption{Comparison result with full analysis and reset-recompute, all times are measured in seconds. The  "Prog" column presents the program. The "I" column presents the commit interval. The "CNode" and "CEdge" columns present changed nodes and changed edges respectively. The "INode" presents the influenced nodes by the changed nodes and edges in reset-recompute strategy. The "RR" presents the reset-recompute strategy. The "BuildCG" presents the time spent on transforming VFG into constraint graph.}
	\label{table:comparefullReset}
	\resizebox{\linewidth}{!}{
\begin{tabular}{l c c c c c c c c}
	\toprule
	Prog & I & CNode & CEdge & Full Time & INodes & RR & BuildCG  & IncSFS \\
	\midrule
	
	\multirow{4}{*}{zstd} & 1 & 0.04\% & 0.00\% & 39.56 (4.16×) & 40.50\% & 29.82 (3.14×) & 4.32 (45.48\%) & 9.50 \\
	& 10 & 0.24\% & 0.01\% & 42.27 (4.02×) & 41.16\% & 30.17 (2.87×) & 3.85 (36.61\%) & 10.51 \\
	& 20 & 0.10\% & 0.01\% & 42.28 (3.85×) & 53.20\% & 38.13 (3.48×) & 5.18 (47.18\%) & 10.97 \\
	& 30 & 0.19\% & 0.04\% & 42.19 (2.38×) & 53.64\% & 39.34 (2.22×) & 5.80 (32.64\%) & 17.76 \\
	\midrule
	\multirow{4}{*}{janet} & 1 & 0.01\% & 0.04\% & 77.46 (9.00×) & 23.72\% & 32.26 (3.75×) & 5.95 (69.17\%) & 8.61 \\
	& 10 & 0.01\% & 0.03\% & 92.31 (4.26×) & 39.62\% & 68.63 (3.17×) & 12.34 (56.90\%) & 21.68 \\
	& 20 & 0.03\% & 0.06\% & 84.71 (5.48×) & 36.00\% & 58.67 (3.80×) & 11.10 (71.83\%) & 15.45 \\
	& 30 & 0.04\% & 0.09\% & 88.83 (2.43×) & 40.25\% & 63.55 (1.74×) & 11.07 (30.23\%) & 36.62 \\
	\midrule
	\multirow{4}{*}{tmux} & 1 & 0.02\% & 0.10\% & 475.84 (6.35×) & 64.62\% & 313.29 (4.18×) & 45.73 (61.04\%) & 74.92 \\
	& 10 & 0.08\% & 0.11\% & 457.38 (4.85×) & 71.41\% & 353.34 (3.75×) & 47.10 (49.99\%) & 94.22 \\
	& 20 & 0.16\% & 0.21\% & 448.74 (3.10×) & 64.99\% & 337.33 (2.33×) & 52.67 (36.39\%) & 144.74 \\
	& 30 & 0.17\% & 0.48\% & 476.62 (3.23×) & 72.32\% & 371.05 (2.52×) & 54.07 (36.65\%) & 147.52 \\
	\midrule
	\multirow{4}{*}{astyle} & 1 & 0.02\% & 0.02\% & 1271.98 (52.19×) & 20.46\% & 351.29 (14.41×) & 11.26 (46.21\%) & 24.37 \\
	& 10 & 0.13\% & 1.93\% & 1288.84 (5.77×) & 61.74\% & 1153.73 (5.16×) & 34.13 (15.27\%) & 223.54 \\
	& 20 & 0.27\% & 1.41\% & 1361.89 (7.07×) & 61.99\% & 1167.50 (6.06×) & 38.91 (20.21\%) & 192.59 \\
	& 30 & 0.29\% & 1.43\% & 1358.15 (7.63×) & 68.99\% & 1288.19 (7.23×) & 29.07 (16.32\%) & 178.07 \\
	\midrule
	\multirow{4}{*}{nginx} & 1 & 0.01\% & 0.01\% & 1387.76 (16.81×) & 68.35\% & 1113.79 (13.49×) & 56.51 (68.46\%) & 82.54 \\
	& 10 & 0.02\% & 0.01\% & 1463.75 (8.57×) & 75.93\% & 1246.75 (7.30×) & 110.09 (64.47\%) & 170.75 \\
	& 20 & 0.06\% & 0.05\% & 1525.32 (17.66×) & 75.96\% & 1329.46 (15.39×) & 43.20 (50.00\%) & 86.40 \\
	& 30 & 0.09\% & 0.06\% & 1529.43 (8.34×) & 75.97\% & 1345.46 (7.33×) & 113.23 (61.73\%) & 183.43 \\
	\midrule
	\multirow{4}{*}{sqlite} & 1 & 0.00\% & 0.01\% & 1670.40 (23.92×) & 30.24\% & 485.39 (6.95×) & 56.73 (81.23\%) & 69.83 \\
	& 10 & 0.01\% & 0.03\% & 1763.95 (11.25×) & 75.81\% & 1225.98 (7.82×) & 110.56 (70.53\%) & 156.76 \\
	& 20 & 0.06\% & 0.07\% & 1775.67 (8.76×) & 68.34\% & 1066.22 (5.26×) & 120.84 (59.61\%) & 202.72 \\
	& 30 & 0.01\% & 0.01\% & 1787.30 (9.42×) & 76.14\% & 1276.73 (6.73×) & 131.42 (69.23\%) & 189.83 \\
	\midrule
	\multicolumn{4}{l}{Avg} & 9.60×  & & 5.84× & 49.89\% & \\
	\bottomrule
\end{tabular}
}

\end{table}

Table~\ref{table:comparefullReset} details the comparison result of total analysis time of \textbf{IncSFS} with \textbf{FullAnalysis} and the \textbf{Reset-Recompute (RR)} strategy. Columns \textbf{CNodes} and \textbf{CEdges} directly reflect the extent of source code modifications, serving as a concrete indicator of the incremental update magnitude. \red{The exact old and new commit IDs and complete results for every sampled pair are archived with the artifact (Section~\ref{sec:data-availability}).}

\textbf{Performance Analysis.}
The \textbf{BuildCG} phase alone  accounts for \avgbuildcg\% of the total IncSFS time. This is primarily we transform the whole VFG to the constraint graph eagerly instead of lazily created if needed, which results in the huge node and edges of the constraint graph and the cost of the initial SCC detection.
Overall, IncSFS achieves a \textbf{\avgfull x} speedup over Full Analysis and a \textbf{\avgreset x} speedup over RR.
The inefficiency of RR is highlighted by the \textbf{INodes} column: simple code changes often impact 20\%--81\% of the VFG nodes due to transitive dependencies. Consequently, RR is forced to recompute the fix-point for a large portion of the program. Our approach only propagates deltas and reuse all old points-to result, resulting in fast reaching a fixed-point.

\subsubsection{Comparison with SILVA}

\begin{table}[t]
	\footnotesize
	\caption{Comparison result with SILVA. The ``Prog'' column presents the program. The ``I'' column presents the commit interval. The ``PNode'' presents the number of processed constraint nodes. The ``SCC'' presents the average number of SCC detections performed on the constraint graph.}
	\label{table:comparewithSILVA}
	\resizebox{\linewidth}{!}{
		\begin{tabular}{l c c c c c c c c}
			\toprule
			\multirow{2}{*}{Prog} & \multirow{2}{*}{I}
			& \multicolumn{3}{c}{IPA}
			& \multicolumn{1}{c}{\red{IncSFS$^{RE-}$}}
			& \multicolumn{3}{c}{IncSFS} \\
			\cmidrule(lr){3-5}
			\cmidrule(lr){6-6}
			\cmidrule(lr){7-9}
			& & Time (s) & PNodes & SCC
			& \red{Time (s)}
			& Time (s) & PNodes & SCC \\
			\midrule
			
			\multirow{4}{*}{janet}
			& 1  & 5.29 (49.8\%) & 200780.9 (47.1\%) & 1.7 (11.8\%) & \red{5.53}  & 2.65 & 106208.2 & 1.5 \\
			& 10 & 14.83 (37.0\%) & 736260.6 (49.7\%) & 18.4 (0.0\%) & \red{14.99} & 9.34 & 370634.4 & 18.4 \\
			& 20 & 9.49 (54.2\%) & 494669.7 (60.0\%) & 2.5 (8.0\%) & \red{10.52} & 4.35 & 197944.3 & 2.3 \\
			& 30 & 40.28 (36.6\%) & 1799976.6 (49.0\%) & 74.7 (0.7\%) & \red{47.02} & 25.55 & 917106.1 & 74.2 \\
			\midrule
			
			\multirow{4}{*}{tmux}
			& 1  & 29.05 (-0.5\%) & 2268396.8 (45.2\%) & 6.1 (3.3\%) & \red{31.22} & 29.19 & 1243195.6 & 5.9 \\
			& 10 & 49.00 (3.8\%) & 3546953.2 (38.9\%) & 123.3 (0.6\%) & \red{54.78} & 47.12 & 2166905.4 & 122.6 \\
			& 20 & 136.54 (32.6\%) & 9830775.9 (59.8\%) & 5.6 (16.1\%) & \red{155.37} & 92.08 & 3953844.5 & 4.7 \\
			& 30 & 134.79 (30.7\%) & 9957423.5 (59.0\%) & 18.4 (5.4\%) & \red{152.93} & 93.46 & 4084256.7 & 17.4 \\
			\midrule
			
			\multirow{4}{*}{zstd}
			& 1  & 4.50 (-15.0\%) & 513078.8 (30.0\%) & 18.0 (1.7\%) & \red{5.06} & 5.18 & 359341.1 & 17.7 \\
			& 10 & 7.28 (8.4\%) & 1264640.5 (57.6\%) & 79.3 (0.5\%) & \red{9.52} & 6.66 & 535809.9 & 78.9 \\
			& 20 & 7.02 (17.5\%) & 1425544.6 (64.4\%) & 20.7 (2.9\%) & \red{8.13} & 5.79 & 506880.3 & 20.1 \\
			& 30 & 13.81 (13.4\%) & 3059401.1 (60.3\%) & 25.5 (2.7\%) & \red{17.58} & 11.96 & 1214429.3 & 24.8 \\
			\midrule
			
			\multirow{4}{*}{astyle}
			& 1  & 19.81 (33.8\%) & 826874.2 (39.4\%) & 17.0 (0.6\%) & \red{20.34} & 13.11 & 500778.3 & 16.9 \\
			& 10 & 278.38 (32.0\%) & 12198190.9 (67.1\%) & 3118.8 (0.0\%) & \red{344.94} & 189.41 & 4013862.5 & 3118.0 \\
			& 20 & 271.89 (43.5\%) & 13051482.3 (75.0\%) & 2591.1 (0.0\%) & \red{362.17} & 153.68 & 3259865.7 & 2590.5 \\
			& 30 & 257.25 (42.1\%) & 12656285.0 (75.0\%) & 2725.7 (0.0\%) & \red{323.86} & 149.00 & 3161002.1 & 2725.0 \\
			\midrule
			
			\multirow{4}{*}{nginx}
			& 1  & 24.66 (-5.6\%) & 400144.5 (50.0\%) & 2.6 (19.2\%) & \red{30.14} & 26.03 & 199902.8 & 2.1 \\
			& 10 & 60.04 (-1.0\%) & 556158.6 (47.7\%) & 28.9 (2.4\%) & \red{71.94} & 60.66 & 290764.6 & 28.2 \\
			& 20 & 39.89 (-8.3\%) & 572395.2 (48.7\%) & 4.0 (12.5\%) & \red{46.91} & 43.19 & 293619.2 & 3.5 \\
			& 30 & 63.74 (-10.1\%) & 832000.8 (50.4\%) & 158.1 (0.5\%) & \red{75.73} & 70.20 & 412325.8 & 157.3 \\
			\midrule
			
			\multirow{4}{*}{sqlite}
			& 1  & 11.80 (-11.1\%) & 395411.8 (36.3\%) & 1.2 (0.0\%) & \red{12.85} & 13.11 & 251969.5 & 1.2 \\
			& 10 & 43.46 (-6.3\%) & 897227.7 (37.1\%) & 943.5 (0.0\%) & \red{48.33} & 46.19 & 563911.6 & 943.4 \\
			& 20 & 82.85 (1.2\%) & 1681167.0 (49.6\%) & 838.6 (0.0\%) & \red{101.51} & 81.88 & 847423.3 & 838.5 \\
			& 30 & 58.93 (0.9\%) & 1817195.9 (48.2\%) & 907.7 (0.0\%) & \red{67.53} & 58.41 & 940493.2 & 907.4 \\
			\midrule
			
			Avg.
			& & 15.8\% & 51.9\% & 3.7\%
			& \red{27.3\%}
			& --
			& --
			& -- \\
			\bottomrule
		\end{tabular}
	}
\end{table}

Table~\ref{table:comparewithSILVA} details the comparison between IncSFS and SILVA. 
We introduce two fine-grained metrics to dissect the performance gains: \textbf{PNodes}  and \textbf{SCC}.
The results indicate that IncSFS is algorithmically superior: it reduces SCC detections by an average of \textbf{\avgscc} and the number of propagated nodes by \textbf{\avgnodes}. Consequently, the overall analysis time is reduced by \textbf{\avgsilva}. The performance gain of IncSFS over SILVA primarily stems from the reduction in SCC detections. 
However, as observed in the results, the absolute reduction in the number of SCCs is not substantial.  Some of the avoided SCCs are small-sized components, which incurs minimal computational overhead on incremental detection, the time savings derived from skipping these operations are consequently moderate. Compared with SILVA, to guarantee that the increase $\Delta^+$ and decrease $\Delta^-$ remain \emph{mutually exclusive} during simultaneous propagation, IncSFS performs additional set difference operations (\emph{i.e.}, $\Delta^+ \setminus \Delta^-$).
In scenarios where updates are naturally exclusive which SILVA handles without extra cost, our method still incurs this calculation overhead, resulting in a slight increase in the average processing time per node.

\red{To further separate the effect of set difference and redundancy elimination, Table~\ref{table:comparewithSILVA} reports an ablation variant, \textbf{$IncSFS^{RE-}$}, which retains the set-difference filtering of IncSFS but disables redundant-propagation avoidance. Compared with $IncSFS^{RE-}$, IncSFS reduces the analysis time by 27.3\%, demonstrating the effectiveness in eliminating redundant propagation, repeated load/store-induced edge updates, and unnecessary SCC maintenance. On the other hand, comparing $IncSFS^{RE-}$ with IPA shows that set-difference filtering introduces an average overhead of 11.5\%.}

\subsubsection{Scalability}
\begin{figure}[t] %
	\centering
	\includegraphics[width=0.9\linewidth]{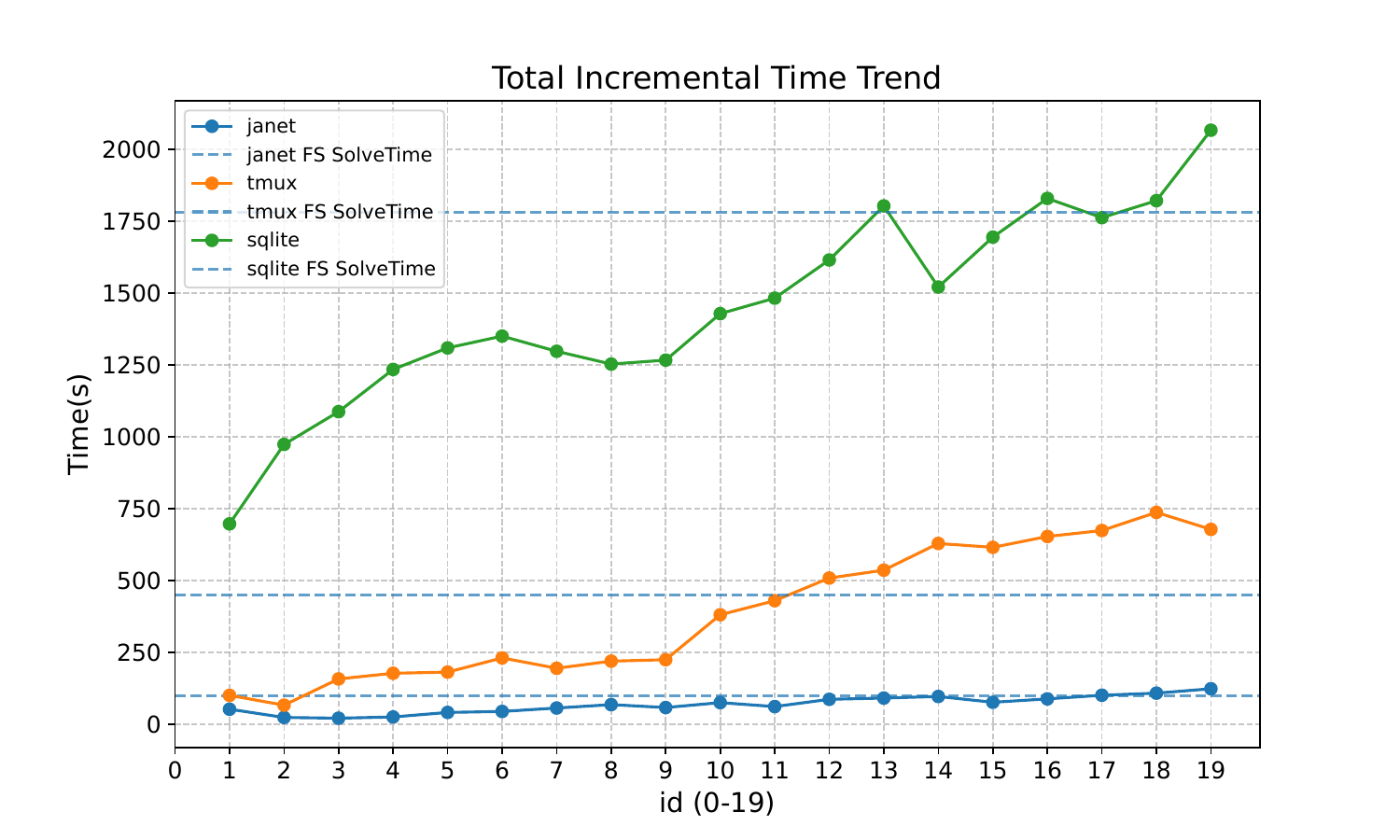}
	\caption{\red{Running with linearly increasing number of commits.}}
	\label{fig:trend}
\end{figure}
\red{Figure~\ref{fig:trend} studies when the incremental overhead exceeds the cost of full re-analysis. The dashed line represents the runtime of full analysis for each program. The results show that IncSFS remains faster than full re-analysis across most modification scales. For insertion-only changes, incremental analysis is always faster than full re-analysis in our experiments. Therefore,  we separately measure the ratios of deleted code and inserted code relative to the original codebase. For \texttt{janet}, \texttt{tmux}, and \texttt{sqlite}, the incremental overhead exceeds full re-analysis only under modification: deleted code accounts for 67.11\%, 33.82\%, and 6.95\% of the original source code size, respectively, while inserted code reaches 150.92\%, 68.67\%, and 5.66\% of the original source code size. These results suggest that the crossover point is mainly caused by large mixed deletion/insertion updates rather than by insertions alone.}

\subsubsection{Correctness Verification}
\begin{wrapfigure}{r}{0.36\linewidth} %
	\centering
	\includegraphics[width=\linewidth]{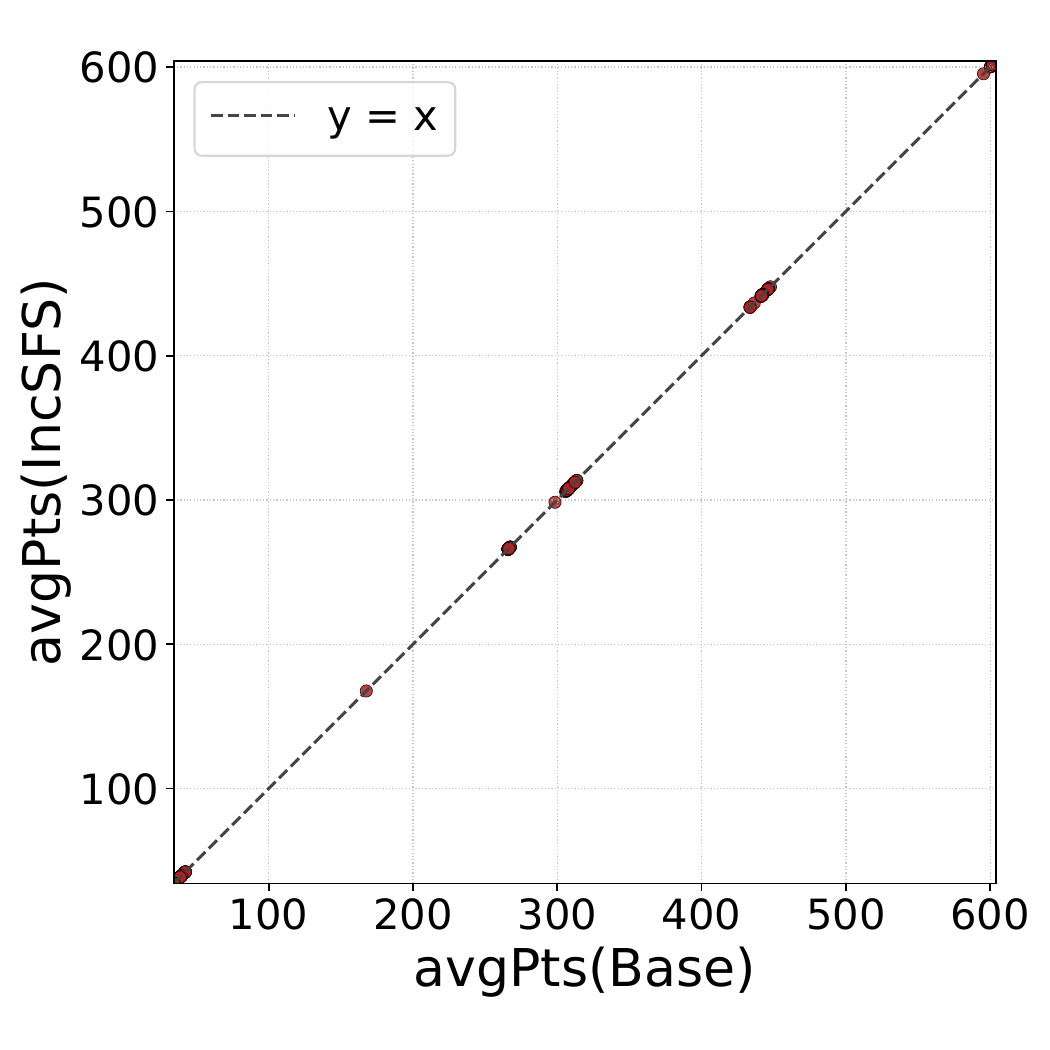}
	\vspace{-5mm} 
	\caption{Avg pts-to set size (IncSFS vs Base).}
	\label{fig:scatter}
\end{wrapfigure}
Figure~\ref{fig:scatter} visualizes the average point-to set size of variables between our incremental approach and the baseline.
As observed, all data points fall precisely on the line $y=x$, which demonstrates that IncSFS produces identical analysis results to the baseline across all tested benchmarks, confirming that our incremental updates  preserve the strict flow-sensitivity of the original analysis. 

\begin{table}[t]
	\centering
	\footnotesize
	\caption{\red{Equivalence between IncSFS and full analysis.}}
	\label{tab:equivalence-counts}
	\resizebox{\linewidth}{!}{
		\begin{tabular}{l r r r}
			\toprule
			\red{Prog} & \red{Version pairs} & \red{Compared variables} & \red{Mismatches} \\
			\midrule
			\red{janet}  & \red{40} & \red{1,766,214} & \red{0} \\
			\red{nginx}  & \red{40} & \red{2,653,743} & \red{0} \\
			\red{tmux}   & \red{40} & \red{2,947,271} & \red{0} \\
			\red{astyle} & \red{40} & \red{2,531,033} & \red{0} \\
			\red{sqlite} & \red{40} & \red{5,582,440} & \red{0} \\
			\red{zstd}   & \red{40} & \red{8,142,750} & \red{0} \\
			\midrule
			\red{Total} & \red{240} & \red{23,623,451} & \red{0} \\
			\bottomrule
		\end{tabular}
	}
\end{table}

\red{Table~\ref{tab:equivalence-counts} reports the number of variables whose points-to sets are compared. Across 23,623,451 variables in 240 configurations, IncSFS produces zero mismatches with full analysis.}

\subsection{Threats to Validity}
\red{We checked for object points-to cycles in \texttt{janet}, \texttt{tmux}, and \texttt{zstd}, but omitted the other programs due to prohibitive memory costs of object point-to cycle detection. Such cycles do occur in most of the checked programs. However, these cycles either fell outside the affected regions or, when cycle-related facts were removed by deletions, still retained independent non-cyclic support. Consequently, no cycle caused IncSFS to reach a non-minimal fixed point in any evaluated configuration, consistent with our exact comparison against full analysis.}

\section{Discussion}

\textbf{Memory usage.} Our rename approach inherently introduces extra constraint nodes resulting in a maximum 3.8$\times$ memory overhead.  By collapsing SCC nodes into a single representative, we significantly reduce computational costs compared to vanilla SFS. This space-for-time trade-off is indispensable for time-cost flow-sensitive analyses. Moreover memory issues  is a problem also faced by other incremental analysis efforts~\cite{10.1145/1069774.1069785,10.1145/3725214} to analyze larger-scale programs. We will investigate reductions in redundant renamed nodes in future work.

\red{\textbf{Downstream client impact.} We have not yet integrated IncSFS into downstream clients such as bug detection. To estimate the potential impact, we use SVF's flow-sensitive pointer-analysis results to construct the value-flow graph and then run the SABER defect detector. In this pipeline, flow-sensitive pointer analysis accounts for 47\% of the total analysis time on average across the benchmark programs evaluated in our experiments. Combining this fraction with IncSFS's average \avgfull x speedup over full analysis, providing an expected end-to-end speedup of 1.73x $(1/((1-0.47)+0.47/\avgfull))$ for the full downstream analysis pipeline.}

\vspace{-2mm}
\section{Related Work}

\textbf{Incremental Program Analysis.}
Incremental analysis has been extensively studied across   dataflow analysis\cite{10.1145/3786763, 10.1145/3720436},abstract interpretation \cite{10.1145/3453483.3454044,10.1145/3648441}, program testing \cite{10.1145/3728883} and regression testing~\cite{8987498}.  Reset-recompute is a widely used strategy in incremental analysis.   Yur \emph{et al.} \cite{841034} introduced the first incremental flow- and context-sensitive pointer analysis using this strategy. However, these methods incur redundant computations limiting the efficiency gain for large-scale pointers analysis as shown in our experiment results.   Lu \emph{et al.} \cite{10.1007/978-3-642-37051-9_4} proposed an incremental flow-insensitive pointer analysis based on CFL-reachability, which incrementally answer queries spanning from variables to corresponding objects.  Krainz \emph{et al.}~\cite{10.1145/3098572.3098578} represents the accesses of a method in a flow-sensitive way, but only handles code insertion.  There are other incremental pointer analysis methods ~\cite{10.1145/3276509,10.1145/3453483.3454026} based on Datalog but hard to leverage the change-local properties based on constraint graph.

To overcome the overhead of recomputation, \textsc{IPA} \cite{10.1145/3293606,10.1145/3725214,10.1145/3527332} proposed efficient incremental algorithms  by propagating deltas in points-to set.  However, these approaches primarily target \emph{flow-insensitive} analysis and fail to exploit the correlation between added and deleted statements~\cite{10.1145/3293606}.  In contrast, our approach targets \emph{flow-sensitive} analysis and exploit a unified way.

\textbf{Flow-Sensitive Pointer Analysis.}
In last decades, various optimizations have been proposed to improve performance of flow-sensitive pointer analysis. Li \emph{et al.} \cite{10.1145/2025113.2025160} converts the flow-sensitive pointer analysis into a graph reachability problem on a value-flow graph. Ben \emph{et al.} \cite{10.1145/1594834.1480911, 5764696} introduced \emph{sparse} flow-sensitive pointer analysis.
\red{Barbar  \emph{et al.} \cite{9370334} merge equivalent points-to sets via object versioning ,  thereby reducing the representational redundancy of classical SFS. IncSFS   dynamically collapses SCCs into representative nodes during incremental solving.} \red{The work proposed by Zhang \emph{et al.}~\cite{ZHANG2026113027} also transforms value-flow information into a constraint graph and improves the efficiency of flow-sensitive pointer analysis in a whole-program setting. Independently and concurrently with their work, our work performs variable renaming in the incremental setting so that points-to deltas can be propagated safely after code changes. }%

\vspace{-2mm}
\section{Conclusion}
We propose the first incremental fully sparse flow-sensitive pointer analysis algorithm by transforming the value-flow graph into a constraint graph and performing SCC detection on that graph.  Furthermore, our method simultaneously propagates both increased and decreased points-to sets, improving analysis efficiency. \red{Under the object-acyclic condition: points-to relation do not formulate a cycle, IncSFS is guaranteed to terminate and compute the least fixed point.} Experiments on six large real-world programs show substantial improvements over both full analysis and the reset--recompute incremental algorithm. In addition, our method delivers higher efficiency than the state-of-the-art incremental pointer analysis algorithm.

\section{Data Availability Statement}
\label{sec:data-availability}
The artifacts and supplementary material are available at \url{https://doi.org/10.5281/zenodo.21771594}~\cite{kunlin_2026_21771594}.

\begin{acks}
This research was supported by the National Key R\&D Program of China (No.~2022YFB4501903) and the NSFC Program (Nos.~62172429 and 62032024).
\end{acks}

\newcounter{combinedbibitem}
\pretocmd{\bibitem}{%
	\stepcounter{combinedbibitem}%
	\ifnum\value{combinedbibitem}=37\balance\fi
}{}{}
\bibliographystyle{ACM-Reference-Format}
\bibliography{sample-base}

\clearpage
\appendix
\section*{Supplementary Material}
This supplement provides complete proofs of termination, constraint consistency, and least-fixed-point convergence. It accompanies the camera-ready paper and is archived with the artifact on \href{https://doi.org/10.5281/zenodo.21771594}{Zenodo}.

\section{Proof of Termination}
\label{sec:supp-termination}

\begin{lemma}
\label{lemma:supp-mon}
	Changes to copy-constraint edges are monotonic with respect to the corresponding points-to-set changes: shrinking a points-to set can only delete edges, whereas expanding it can only insert edges.
\end{lemma}

\begin{proof}
	
	We consider shrinking and expanding $\pts(p)$.
	
	\begin{itemize}
		
		\item \textbf{Shrinkage.}
		
		For a store $\ell:*p=q$, removing objects from $\pts(p)$ deletes the corresponding target edges. A transition from a weak to a strong update additionally deletes the bypass edges for the unique target, while a transition to a no-op removes all edges induced by the store. No new edge is introduced.
		
		For a load $\ell:p=*q$, removing $o$ from $\pts(q)$ deletes every induced edge $o_{\ell^\prime}\rightarrow p$ with $\ell^\prime\xrightarrow{o}\ell$.
		
		\item \textbf{Expansion.}
		For a store $\ell:*p=q$, adding objects to $\pts(p)$ inserts the corresponding target edges. A transition from a strong to a weak update restores bypass edges, while a transition from a no-op introduces the edges required by the resulting strong or weak update. No edge is deleted.
		
		For a load $\ell:p=*q$, adding $o$ to $\pts(q)$ inserts every induced edge $o_{\ell^\prime}\rightarrow p$ with $\ell^\prime\xrightarrow{o}\ell$.
		
	\end{itemize}
\end{proof}

\begin{lemma}
	The deletion of copy constraint edges can only lead to a reduction of points-to sets, while the addition of copy constraint edges can only lead to an expansion of points-to sets.
\end{lemma}

\begin{proof}
	Deleting an edge $p\rightarrow q$ seeds $\Delta^-_q$ with $\pts(p)$. After predecessor filtering, the remaining negative delta can only remove objects from $\pts(q)$; by Lemma~\ref{lemma:supp-mon}, the resulting load/store side effects are further edge deletions. Conversely, inserting $p\rightarrow q$ seeds $\Delta^+_q$ with $\pts(p)$, which can only enlarge $\pts(q)$ and induce further edge insertions.
\end{proof}

\begin{definition}[Object points-to cycle]
\label{def:supp-object-pts-cycle}
For object nodes $o,o'$, write $o\leadsto_{\mathrm{pt}}o'$ iff $o'\in\pts(o)$. An object points-to cycle is a nonempty sequence $o_1,\ldots,o_n$ such that $o_i\leadsto_{\mathrm{pt}}o_{i+1}$ for $1\le i<n$ and $o_n\leadsto_{\mathrm{pt}}o_1$. A points-to relation is \emph{object-acyclic} iff it contains no such cycle.
\end{definition}

\begin{theorem}
If object-acyclicity is preserved throughout the analysis, the interleaved algorithm terminates.
\end{theorem}

\begin{proof}
	We prove termination by contradiction. Assume that the algorithm does not terminate.
	Then there must exist an infinite sequence of insertions and deletions of some constraint edge $x \rightarrow y$. 
	Indeed, if every edge were modified only finitely many times, then the algorithm would eventually reach a fixed point and terminate.
	We analyze all possible cases.
	
	\paragraph{Case 1: alternating global addition and deletion.}
	
	Assume that one iteration performs only edge insertions, while the next iteration performs only edge deletions.
	This situation is impossible.
	In our framework, inserting constraint edges can only enlarge points-to sets. 
	Moreover, enlarging points-to sets can only trigger additional edge insertions, rather than deletions, by Lemma~\ref{lemma:supp-mon}.
	Therefore, an iteration consisting solely of deletions cannot immediately follow an iteration consisting solely of insertions. Hence, infinite oscillation cannot arise from globally alternating insertion-only and deletion-only iterations.
	
	\paragraph{Case 2: simultaneous insertion and deletion within one iteration.}
	
	The remaining possibility is that, within the same iteration, some edges are inserted while others are deleted.
	Let
	\[
	S_k=\langle G_k,\mathcal{E}_k^-,\mathcal{E}_k^+\rangle
	\]
	be the state after round $k$. Here, $G_k$ is the constraint graph annotated with the current points-to set of every node, while $\mathcal{E}_k^-$ and $\mathcal{E}_k^+$ are the pending edge deletions and insertions derived by \textsc{P-Load} and \textsc{P-Store}. Since the program contains finitely many variables, abstract objects, and possible constraint edges, the set of states is finite.

	Assume, for contradiction, that the algorithm does not terminate. An infinite execution over this finite state space contains two equal round-end states: $S_i=S_j$ for some $i<j$. Define
	\[
	D_{i,j}^- = \bigcup_{i\le m<j}\mathcal{E}_m^- ,
	\qquad
	D_{i,j}^+ = \bigcup_{i\le m<j}\mathcal{E}_m^+ .
	\]
	The filtering rules enforce $\mathcal{E}_m^-\cap\mathcal{E}_m^+=\emptyset$ in every round. Therefore, an edge cannot be deleted and reinserted in the same round. Because $S_i=S_j$, every edge deleted between rounds $i$ and $j$ must be reinserted in a later round before $j$, and every inserted edge must similarly be deleted. Hence,
	\[
	D_{i,j}^-=D_{i,j}^+\neq\emptyset.
	\]
	If this set were empty, no derived edge or points-to set would change after $S_i$, and the algorithm would already have terminated.

	Let $e\Rightarrow e'$ mean that the points-to reduction caused by derived-edge deletion $e$ triggers deletion $e'$ through \textsc{P-Load} or \textsc{P-Store}. Since every update in $D_{i,j}^-$ must be reversed before $S_j$, following these causes yields a cycle $e_1\Rightarrow e_2\Rightarrow\cdots\Rightarrow e_n\Rightarrow e_1$. Let $e_r(1\leq r\leq n)$ be induced when the store $*p_r=a_r$ loses target object $o_r$. If $e_r$ causes $e_{r+1}$, the reduction propagated through $o_r$ removes $o_{r+1}$ at the next store; thus $o_{r+1}\in\pts(o_r)$, with indices modulo $n$. Hence,
	$
	o_2\in\pts(o_1),\ o_3\in\pts(o_2),\ldots,\ o_1\in\pts(o_n),
	$
	which gives $o_1\leadsto_{\mathrm{pt}}o_2\leadsto_{\mathrm{pt}}\cdots\leadsto_{\mathrm{pt}}o_n\leadsto_{\mathrm{pt}}o_1$.
	This contradicts the acyclicity assumption. Therefore, the algorithm terminates.
\end{proof}

\section{Proof of Consistency and Least-Fixed-Point Convergence}
\label{sec:supp-consistency}
\begin{theorem}[Constraint consistency]
	\label{thm:constraint-consistency}
	Under object-acyclicity, the constraint graph maintained by IncSFS is identical to the graph obtained by transforming the updated program, and the points-to sets produced by IncSFS form a fixed point of this updated constraint system.
\end{theorem}
\begin{proof}
	For each statement type in the updated program, we show that the points-to sets and constraint graph maintained by IncSFS satisfy the data-flow functions of full-sparse flow-sensitive analysis.
	\begin{itemize}
		
		\item \textbf{Alloc statement.} Consider a copy statement $\ell:p=alloc$. We show that $o\in \pts(p)$. Assume $o\not\in \pts(p)$, which is invalid since transform rule $o\rightarrow \pts(p)$.
		
		\item \textbf{Copy statement.} Consider a copy statement $\ell:p=q$ (the same argument applies to field statements and phi statements). We show that $\forall o\in\pts(q), o\in\pts(p)$. Assume, for contradiction, that $\exists o\in\pts(q)\wedge o\notin\pts(p)$. Then either $o\in\Delta^-_p$ or $o\notin\Delta^+_p$.
		
		If $o\in\Delta^-_p$, since there exists a copy edge $q\rightarrow p$, regardless of whether this edge is newly inserted or already exists, the propagation rule applies:
		\[
		\Delta^-_p \leftarrow \Delta^-_p \setminus \pts(q).
		\]
		Because $o\in\pts(q)$, $o$ must be removed from $\Delta^-_p$. Therefore this case is invalid.
		
		If $o\notin\Delta^+_p$, we further distinguish two cases. First, if initially $o\in\Delta^+_p$, then by the filtering rule
		\[
		\Delta^+_p \leftarrow \Delta^+_p \setminus \pts(p),
		\]
		$o$ can only be removed when $o\in\pts(p)$, contradiction. Otherwise, initially $o\notin\Delta^+_p$. Then either $o\in\Delta^+_q$ but $o\notin\Delta^+_p$, which is impossible according to Rule \textsc{Propagate} $\Delta^+_p\leftarrow \Delta^+_p\cup\Delta^+_q$, or the edge $q\rightarrow p$ is newly inserted, in which case Rule \textsc{insEdge} requires $\Delta^+_p\leftarrow\Delta^+_p\cup\pts(q)$. These two cases are also invalid. Therefore, $\forall o\in\pts(q), o\in\pts(p)$.
		
		\item \textbf{Store statement.} Consider a store statement $\ell:*p=q$. We first show that $\forall o\in\pts(p), q\rightarrow o_\ell$. Assume, for contradiction, that there exists $o\in\pts(p)$ such that the edge $q\rightarrow o_\ell$ does not exist. According to Rule \textsc{P-Store}, this can only happen if either $o\in\Delta^-_p$ or $o\notin\Delta^+_p$. However, by construction, $\pts(p)\cap\Delta^-_p=\emptyset$, thus $o\in\Delta^-_p$ is impossible.
		
		Therefore, only the case $o\notin\Delta^+_p$ remains. This implies that the edge $p\rightarrow o_\ell$ did not exist before and $o\in\pts(p)$ all the time. The statement must previously correspond either to a weak update or a strong update. In both cases, according to the transform rules, the edge $q\rightarrow o_\ell$ must necessarily exist. This is invalid.
		
		We now consider bypass edges. For weak updates, we require $\forall \ell'\xrightarrow{o}\ell, o_{\ell'}\rightarrow o_\ell$. Assume such an edge does not exist. According to Rule \textsc{P-Store}, the edge will be inserted through
		$
		\mathcal{E}^+ \leftarrow \mathcal{E}^+ \cup \{o_{\ell'}\rightarrow o_\ell\},
		$
		This case is invalid.
		
		For strong updates, we require $\forall \ell'\xrightarrow{o}\ell, o\notin\pts(p)\Rightarrow o_{\ell'}\rightarrow o_\ell$. Assume such an edge does not exist. Then Rule \textsc{P-Store} will insert the edge. This case is invalid. Conversely, if $o\in\pts(p)$, then the bypass edge must not exist. Assume it exists. Then Rule \textsc{P-Store} deletes the edge. This case is invalid.
		
		For no-op updates, all bypass edges must be absent, i.e., $\forall \ell'\xrightarrow{o}\ell, o_{\ell'}\not\rightarrow o_\ell$. Assume such an edge exists. Then Rule \textsc{P-Store} will delete the edge. This case is invalid.
		
		\item \textbf{Load statement.} Consider a load statement $\ell:p=*q$. We show that $\forall \ell'\xrightarrow{o}\ell, o\in\pts(q)\Rightarrow o_{\ell'}\rightarrow o_\ell$. Assume, for contradiction, that $o\in\pts(q)$ but the edge $o_{\ell'}\rightarrow o_\ell$ does not exist. Then either $o\in\Delta^-_q$ or $o\notin\Delta^+_q$.
		
		If $o\in\Delta^-_q$, this is impossible because deleted objects cannot simultaneously remain in $\pts(q)$. Otherwise, if $o\notin\Delta^+_q$, then $o$ has remained continuously in $\pts(q)$ throughout the incremental analysis. According to transform rule \textsc{Load}, the edge $o_{\ell'}\rightarrow o_\ell$ must therefore already exist, contradiction.
		
	\end{itemize}
\end{proof}

\begin{lemma}[Incoming Neighbors Property]
	\label{lemma:incoming-neighbors}
	As established by Liu et al.~\cite{10.1145/3293606}, consider an acyclic PAG and a pointer node $q$ such that $o\in\pts(q)$. If $q$ has an incoming neighbor $r$, i.e., there exists an edge $r\rightarrow q$, and $o\in\pts(r)$, then there exists a path from $o$ to $r$ that does not pass through $q$.
\end{lemma}

\begin{theorem}[Conditional Change Local Property]
	\label{thm:incoming-local-test}
	Under the object-acyclicity condition, suppose that an edge $p\rightarrow q$ is deleted from an acyclic PAG and all
	other edges remain unchanged. For any $o\in\pts(q)$, if $q$ has an incoming
	neighbor $r$ such that $o\in\pts(r)$, then $o$ remains in $\pts(q)$. If no
	incoming neighbor of $q$ has a points-to set containing $o$, then $o$ must be
	removed from $\pts(q)$.
\end{theorem}

\begin{proof}
\begin{itemize}

	\item Case 1: If $o\in\pts(p)$ and there exists an incoming neighbor $r$ of $q$ such that the edge $r\rightarrow q$ is not deleted when $o$ is removed from $\pts(q)$. 
	By Lemma~\ref{lemma:incoming-neighbors}, $o$ reaches $r$ along a path that does not pass through $q$ and this path will not be unreachable caused by the condition above.  This path still supports $o\in\pts(q)$. Therefore, $o$ remains in $\pts(q)$ after deleting $p \rightarrow q$. Otherwise, if no incoming neighbor of $q$ has a points-to set containing $o$, then $o$ cannot reach $q$ and should therefore be removed from $\pts(q)$.
	
	\item Case 2: If $o\in\pts(p)$ and there exists an incoming neighbor $r$ of $q$ such that the edge $r\rightarrow q$ is deleted when $o$ is removed from $\pts(q)$. The following program provides an example for this case. This case violates the object-acyclicity condition. If removing $o$ must eventually delete $r\rightarrow p$. There must exist a chain $e_1\Rightarrow e_2\Rightarrow\cdots\Rightarrow e_n$, where $e_n=(r\rightarrow p)$. Each step is induced by a store $*p_k=a_k$: losing target $o_k$ deletes $a_k\rightarrow o_k$ and causes $o_k$ to lose $o_{k+1}$, which implies
	$
	o_1\in\pts(p), o_2\in\pts(o_1),\ldots, p\in\pts(o_{n-1}),
	$
	which contradicts object-acyclicity.

\end{itemize}
\end{proof}

\begin{figure}[t]
	\centering
	\begin{minipage}[t]{0.46\linewidth}
		\vspace{0pt}
		\begin{lstlisting}[
			basicstyle=\ttfamily\small,
			columns=fullflexible,
			frame=none
			]
			p  = alloc o1;
			q0 = alloc o2;
			t  = p;
			u  = p;
			q0 = p;  // deleted
			q  = q0;
			label:
			q.f = t;
			r   = u.f;
			q   = r;
			goto label;
		\end{lstlisting}
	\end{minipage}
	\hfill
	\begin{minipage}[t]{0.50\linewidth}
		\vspace{0pt}
		\centering
	\includegraphics[width=\linewidth]{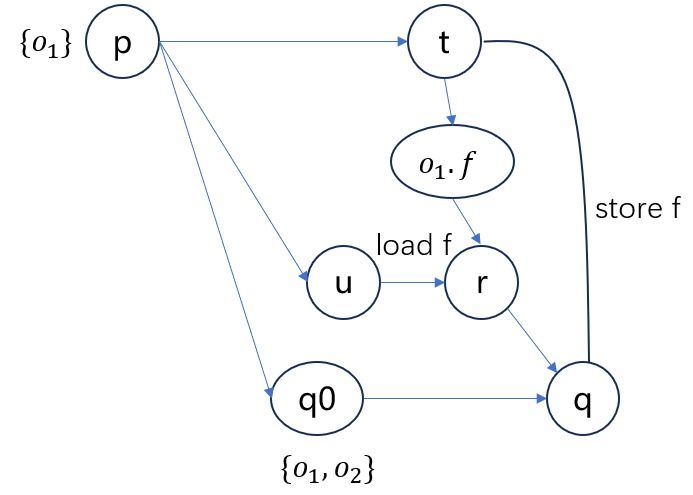}
	\end{minipage}
	
	\caption{An example program and its corresponding graph.}
	\Description{An example pointer program and its points-to dependency graph, illustrating a cyclic dependency among pointer facts.}
	\label{fig:example-program-graph}
\end{figure}

We prove for minimality according to Theorem \ref{thm:incoming-local-test}. Assume that $o\in\pts(p)$, but $o\notin\pts(q)$ for every incoming edge $q\rightarrow p$. Since both $\Delta^-$ and $\Delta^+$ are derived from predecessor nodes, the presence of $o$ in $\pts(p)$ must have originated from some predecessor $r\rightarrow p$. There are three cases: Case 1 : $o\in\Delta^-_r$ and $r\rightarrow p$ still exists; Case 2: the edge $r\rightarrow p$ is deleted and $o\in\pts(r)$; or Case 3: $q\rightarrow p$ is a newly inserted edge used to filter the removal of $o$ in $\pts(p)$ and $o\in\pts(q)$.
 By Theorem~\ref{thm:incoming-local-test}, Case~3 is impossible under object-acyclicity. For Cases~1 and~2, since no current predecessor of $p$ contains $o$, Rule \textsc{Filter-Del} ensures that $o$ is eventually removed from $\pts(p)$, contradicting the assumption that $o\in\pts(p)$. Therefore, IncSFS reaches the least fixed point.

\end{document}